\documentclass[11pt]{article}

\usepackage{amsmath}
\usepackage{amssymb}
\usepackage{amsthm}
\usepackage{fullpage}
\usepackage{rotating}
\usepackage{hyperref}
\usepackage{enumitem}
\usepackage{tikz}
\usepackage{bm}
\usepackage{xcolor}

\newtheorem{theorem}{Theorem}

\newtheorem{lemma}[theorem]{Lemma}
\newtheorem{observation}[theorem]{Observation}

\newtheorem{definition}[theorem]{Definition}
\newtheorem{claim}[theorem]{Claim}
\newtheorem{fact}[theorem]{Fact}
\newtheorem{remark}[theorem]{Remark}

\makeatletter
\newcommand{\newreptheorem}[2]{\newtheorem*{rep@#1}{\rep@title}\newenvironment{rep#1}[1]{\def\rep@title{#2 \ref*{##1}}\begin{rep@#1}}{\end{rep@#1}}}
\makeatother
\newreptheorem{theorem}{Theorem}
\newreptheorem{proposition}{Proposition}
\newreptheorem{lemma}{Lemma}

\newcommand{\prob}[2]{\mathop{\mathrm{Pr}}_{#1}[#2]}
\newcommand{\avg}[2]{\mathop{\textbf{E}}_{#1}\left[#2\right]}
\newcommand{\poly}{\mathop{\mathrm{poly}}}

\newcommand{\F}{\mathbb{F}}

\newcommand{\ceil}[1]{\ensuremath{\left\lceil #1 \right\rceil}}

\newcommand{\mc}[1]{\mathcal{#1}}

\newcommand{\Maj}{\mathrm{Maj}}

\newcommand{\AC}{\mathrm{AC}}

\newcommand{\Vars}{\mathrm{Vars}}

\usepackage{tikz}

\newcommand{\comment}[2]{}

\newcommand{\nutan}[1]{\comment{N}{#1}}
\newcommand{\srikanth}[1]{\comment{S}{#1}}

\title{A Fixed-Depth Size-Hierarchy Theorem for $\AC^0[\oplus]$ via the Coin Problem}
\author{Nutan Limaye
\thanks{Indian Institute of Technology, Bombay. Email: \texttt{nutan@cse.iitb.ac.in}} 
\and
 Karteek Sreenivasaiah
 \thanks{Indian Institute of Technology Hyderabad. Email: \texttt{karteek@iith.ac.in}}
 \and
 Srikanth Srinivasan
 \thanks{Indian Institute of Technology, Bombay.  Email: \texttt{srikanth@math.iitb.ac.in}} 
 \and
 Utkarsh Tripathi 
 \thanks{Indian Institute of Technology, Bombay. Supported by the Ph.D. Scholarship of NBHM, DAE, Government of India.  Email: \texttt{utkarshtripathi.math@gmail.com}}
 \and
 S. Venkitesh
 \thanks{Indian Institute of Technology, Bombay.  Supported by the Senior Research Fellowship of HRDG, CSIR, Government of India.  Email: \texttt{venkitesh.mail@gmail.com}}
 }
\begin{document}

\maketitle
\thispagestyle{empty}

\begin{abstract}

In this work we prove the first \emph{Fixed-depth Size-Hierarchy Theorem} for uniform  $\AC^0[\oplus]$. In particular, we show that for any fixed $d$, the class $\mc{C}_{d,k}$ of functions that have uniform $\AC^0[\oplus]$ formulas of depth $d$ and size $n^k$ form an infinite hierarchy. We show this by exibiting the first class of \emph{explicit} functions where we have nearly (up to a polynomial factor) matching upper and lower bounds for the class of $\AC^0[\oplus]$ formulas.

The explicit functions are derived from the \emph{$\delta$-Coin Problem}, which is the computational problem of distinguishing between coins that are heads with probability $(1+\delta)/2$ or $(1-\delta)/2,$ where $\delta$ is a parameter that is going to $0$. We study the complexity of this problem and make progress on both upper bound and lower bound fronts. 
\begin{itemize}
\item \textbf{Upper bounds.} For any constant $d\geq 2$, we show that there are \emph{explicit} monotone $\AC^0$ formulas (i.e. made up of AND and OR gates only) solving the $\delta$-coin problem that have depth $d$, size $\exp(O(d(1/\delta)^{1/(d-1)}))$, and sample complexity (i.e. number of inputs) $\poly(1/\delta).$ This matches previous upper bounds of O'Donnell and Wimmer (ICALP 2007) and Amano (ICALP 2009) in terms of size (which is optimal) and improves the sample complexity  from $\exp(O(d(1/\delta)^{1/(d-1)}))$ to $\poly(1/\delta)$.
\item \textbf{Lower bounds.} We show that the above upper bounds are nearly tight (in terms of size) even for the significantly stronger model of $\AC^0[\oplus]$ formulas (which are also allowed NOT and Parity gates): formally, we show that any $\AC^0[\oplus]$ formula solving the $\delta$-coin problem must have size $\exp(\Omega(d(1/\delta)^{1/(d-1)})).$ This strengthens  a result of Shaltiel and Viola (SICOMP 2010), who prove a $\exp(\Omega((1/\delta)^{1/(d+2)}))$ lower bound for $\AC^0[\oplus]$, and a result of Cohen, Ganor and Raz (APPROX-RANDOM 2014), who show a $\exp(\Omega((1/\delta)^{1/(d-1)}))$ lower bound for the smaller class $\AC^0$.
\end{itemize}

The upper bound is a derandomization involving a use of Janson's inequality (from probabilistic combinatorics) and classical combinatorial designs; as far as we know, this is the first such use of Janson's inequality. For the lower bound, we prove an optimal (up to a constant factor) degree lower bound for multivariate polynomials over $\F_2$ solving the $\delta$-coin problem, which may be of independent interest.
\end{abstract}

\newpage

\section{Size-Hierarchy theorems for $\AC^0[\oplus]$}
\label{sec:sht}
\setcounter{page}{1}

Given any natural computational resource, an intuitive conjecture one might make is that access to more of that resource results in more computational power. \emph{Hierarchy theorems} make this intuition precise in various settings. Classical theorems in Computational complexity theory such as the time and space hierarchy theorems~\cite{Time-hie,space-hie,NTime-hie,BPP-hie} show that Turing Machine-based computational models do become strictly more powerful with more access to time or space respectively.

The analogous questions in the setting of Boolean circuit complexity deal with the complexity measures  of depth and size of Boolean circuits. Both of these have been intensively studied in the case of $\AC^0$ circuits and by now, we have near-optimal \emph{Depth} and \emph{Size-hierarchy} theorems for $\AC^0$ circuits~\cite{Hastad,Rossman,Amano-SHT,HRST}.  Our focus in this paper is on size-hierarchy theorems for $\AC^0[\oplus]$ circuits. %

Essentially, a size-hierarchy theorem for a class of circuits says that there are Boolean functions $f:\{0,1\}^n\rightarrow \{0,1\}$  that can be computed by circuits of some size $s = s(n)$ but not by circuits of size significantly smaller than $s$, e.g. $\sqrt{s}$. However, stated in this way, such a statement is trivial to prove, since we can easily show by counting arguments that there are more functions computed by circuits of size $s$ than by circuits of size $\sqrt{s}$ and hence there must be a function that witnesses this separation. As is standard in the setting of circuits, what is interesting is an \emph{explicit} separation. (Equivalently, we could consider the question of separating uniform versions of these circuit classes.) 

Strong results in this direction are known in the setting of $\AC^0$ circuits (i.e. constant-depth Boolean circuits made up of AND, OR and NOT gates).

\paragraph{Size hierarchy theorem for $\AC^0$.} In order to prove a size-hierarchy theorem for $\AC^0,$ we need an explicit function that has circuits of size $s$ but no circuits of size less than $\sqrt{s}.$ If we fix the depth of the circuits under consideration, a result of this form follows immediately from the \emph{tight} exponential $\AC^0$  circuit lower bound of H\r{a}stad~\cite{Hastad} from the 80s. H\r{a}stad shows that any depth-$d$ $\AC^0$ circuit for the Parity function on $n$ variables must have size $\exp(\Omega(n^{1/(d-1)}));$ further, this result is tight, as demonstrated by a folklore depth-$d$ $\AC^0$ upper bound of $\exp(O(n^{1/(d-1)}))$. Using both the lower and upper bounds for Parity, we get a separation between circuits of size $s_0 = \exp(O(n^{1/(d-1)}))$ and $s_0^{\varepsilon}$ for some fixed $\varepsilon > 0.$ The same separation also holds between $s$ and $s^\varepsilon$ for any $s$ such that $s\leq s_0,$ since we can always take the Parity function on some $m\leq n$ variables so that the above argument goes through. We thus get a \emph{Fixed-depth Size-Hierarchy theorem} for $\AC^0$ for any $s(n)\leq \exp(n^{o(1)}).$

Even stronger results are known for $\AC^0.$ Note that the above results do not separate, for example, size $s$ circuits of depth $2$ from size $s^{\varepsilon}$ circuits of depth $3$. However, recent results of Rossman~\cite{Rossman-clique} and Amano~\cite{Amano-SHT} imply the existence of explicit functions\footnote{The explicit functions are the $k$-clique problem and variants.} that have $\AC^0$ circuits of depth $2$ and size $n^k$ (for any constant $k$) but not $\AC^0$ circuits of \emph{any} constant depth and size $n^{k-\varepsilon}.$

\paragraph{$\mathbf{\AC^0[\oplus]}$ setting.} Our aim is to prove size-hierarchy theorems for $\AC^0[\oplus]$ circuits (i.e. constant-depth Boolean circuits made up of AND, OR, NOT and $\oplus$ gates).\footnote{Our results also extend straightforwardly to $\AC^0[\mathrm{MOD}_p]$ gates for any constant prime $p$ (here, a $\mathrm{MOD}_p$ gate accepts if the sum of its input bits is non-zero modulo $p$).} As for $\AC^0,$ we have known exponential lower bounds for this circuit class from the 80s, this time using the results of Razborov~\cite{Razborov} and Smolensky~\cite{Smolensky}. Unfortunately, however, most of these circuit lower bounds are \emph{not} tight. For instance, we know that the Majority function on $n$ variables does not have $\AC^0[\oplus]$ circuits of depth $d$ and size $\exp(\Omega(n^{1/2(d-1)})),$ but the best upper bounds are worse than $\exp(O(n^{1/(d-1)}))$ (in fact, the best known upper bound~\cite{Maj-ubd} is an $\AC^0$ circuit of size $\exp(O(n^{1/(d-1)}) (\log n)^{1-1/(d-1)})$).\nutan{Is this correct?}\footnote{A similar fact is also true for the $\mathrm{MOD}_p$ functions, for $p$ an odd prime.} As a direct consequence of this fact, we do not even have \emph{fixed-depth} size-hierarchy theorems of the above form for $\AC^0[\oplus]$: the known results only yield a separation between circuits of size $s$ and circuits of size $\exp(O(\sqrt{\log s}))$, which is a considerably worse result.

In this work we present the first fixed-depth size-hierarchy theorem for $\AC^0[\oplus]$ circuits. Formally, we prove the following.

\begin{theorem}[Fixed-depth size-hierarchy theorem]
\label{thm:size-hie}
There is an absolute constant $\varepsilon_0 \in (0,1)$ such that the following holds. For any fixed depth $d\geq 2$, and for infinitely many $n$ and any $s = s(n) = \exp(n^{o(1)}),$ there is an explicit monotone depth-$d$ $\AC^0$ formula $F_n$ on $n$ variables of size at most $s$ such that any $\AC^0[\oplus]$ formula computing the same function has size at least $s^{\varepsilon_0}.$

In particular, if $\mc{C}_{d,k}$ denotes the family of languages that have uniform $\AC^0[\oplus]$ formulas of depth $d$ and size $n^k$, then the hierarchy $\mc{C}_{d,1}\subseteq \mc{C}_{d,2}\cdots$ is infinite.
\end{theorem}

We can also get a similar result for $\AC^0[\oplus]$ \emph{circuits} of fixed depth $d$ by using the fact that circuits of depth $d$ and size $s_1$ can be converted to formulas of depth $d$ and size $s_1^{d}.$ Using this idea, we can get a separation between circuits (in fact formulas) of depth $d$ and size $s$ and circuits of depth $d$ and size $s^{\varepsilon_0/d}.$

To get this (almost) optimal fixed-depth size-hierarchy theorem we design an explicit function $f$ and obtain tight upper and lower bounds for it for each fixed depth $d$.  The explicit function is based on the \emph{Coin Problem}, which we define below. 

\section{The Coin Problem}

The Coin Problem is the following natural computational problem. Given a two-sided coin that is heads with probability either $(1+\delta)/2$ or $(1-\delta)/2$, decide which of these is the case. The algorithm is allowed many independent tosses of the coin and has to accept in the former case with probability at least $0.9$ (say) and accept in the latter case with probability at most $0.1.$ The formal statement of the problem is given below.

\begin{definition}[The $\delta$-Coin Problem] 
\label{def:coinproblem}
For any $\alpha\in [0,1]$ and integer $N\geq 1$, let $D_{\alpha}^N$ be the product distribution over $\{0,1\}^N$ obtained by setting each bit to $1$ independently with probability $\alpha.$

Let $\delta \in (0,1)$ be a parameter. Given an $N\in \mathbb{N},$ we define the probability distributions $\mu_{\delta,0}^N$ and $\mu_{\delta,1}^N$ to be the distributions $D^N_{(1-\delta)/2}$ and $D^N_{(1+\delta)/2}$ respectively. We omit the $\delta$ in the subscript when it is clear from context.

Given a function $g: \{0,1\}^N \rightarrow \{0,1\}$, we say that \emph{$g$ solves the $\delta$-coin problem with error $\varepsilon$} if 
\begin{equation}
\label{eq:cpdefn}
\prob{\bm{x}\sim \mu_0^N}{g(\bm{x}) = 1} \leq \varepsilon \text{  and   } \prob{\bm{x}\sim \mu_1^N}{g(\bm{x}) = 1} \geq 1-\varepsilon.
\end{equation}
In the case that $g$ solves the coin problem with error $0.1,$ we simply say that $g$ solves the $\delta$-coin problem (and omit mention of the error).

We say that the \emph{sample complexity} of $g$ is $N$.
\end{definition}

We think of $\delta$ as a parameter that is going to $0$. We are interested in both the sample complexity and \emph{computational complexity} of functions solving the coin problem. Both these complexities are measured as a function of the parameter $\delta.$

The problem is folklore, and has also been studied (implicitly and explicitly) in many papers in the Computational complexity literature~\cite{ABO,OW,Amano,SV,BV,Viola,Steinberger,CGR,LV,RossmanS}. It was formally introduced in the work of Brody and Verbin~\cite{BV}, who studied it with a view to devising pseudorandom generators for Read-once Branching Programs (ROBPs).

It is a standard fact that $\Omega(1/\delta^2)$ samples are necessary to solve the $\delta$-coin problem (irrespective of the computational complexity of the underlying function). Further, the algorithm that takes $O(1/\delta^2)$ many independent samples and accepts if and only if the majority of the coin tosses are heads, does indeed solve the $\delta$-coin problem. We call this the ``trivial solution'' to the coin problem.

It is not clear, however, if this is the most computationally ``simple'' method of solving the coin problem. Specifically, one can ask if the $\delta$-coin problem can be solved in computational models that cannot compute the Boolean Majority function on $O(1/\delta^2)$ many input bits. (Recall that the Majority function on $n$ bits accepts inputs of Hamming weight greater than $n/2$ and rejects other inputs.)

Such questions have received quite a bit of attention in the computational complexity literature. Our focus in this paper is on the complexity of this problem in the setting of $\AC^0$  and $\AC^0[\oplus]$ circuits.

Perhaps surprisingly, the Boolean circuit complexity of the coin problem in the above models is \emph{not} the same as the circuit complexity of the Boolean Majority function. We describe below some of the interesting upper as well as lower bounds known for the coin problem in the setting of constant-depth Boolean circuits. 

\paragraph{Known upper bounds.} It is implicit in the results of O'Donnell and Wimmer~\cite{OW} and Amano~\cite{Amano} (and explicitly noted in the paper of Cohen, Ganor and Raz~\cite{CGR}) that the complexity of the coin problem is closely related to the complexity of \emph{Promise} and \emph{Approximate} variants of the Majority function. Here, a \emph{promise majority} is a function that accepts inputs of relative Hamming weight at least $(1+\delta)/2$ and rejects inputs of relative Hamming weight at most $(1-\delta)/2;$ and an \emph{approximate majority} is a function that agrees with the Majority function on $90\%$ of its inputs.\footnote{Unfortunately, both these variants of the Majority function go by the name of ``approximate majority'' in the literature.}

O'Donnell and Wimmer~\cite{OW} and Amano~\cite{Amano} show that the $\AC^0$ circuit complexity of some approximate majorities is superpolynomially \emph{smaller} than the complexity of the Majority function. More specifically, the results in~\cite{OW,Amano} imply that there are approximate majorities that are computed by monotone $\AC^0$ \emph{formulas} of depth $d$ and size $\exp(O(dn^{1/2(d-1)}))$, while a well-known result of H\r{a}stad~\cite{Hastad} implies that any $\AC^0$ circuit of depth $d$ for computing the Majority function must have size $\exp(\Omega(n^{1/(d-1)})).$ For example, when $d=2$, there are approximate majorities that have formulas of size $\exp(O(\sqrt{n}))$ while any circuit for the Majority function must have size $\exp(\Omega(n)).$ These upper bounds were slightly improved to $\AC^0$ \emph{circuits} of size $\exp(O(n^{1/2(d-1)}))$ in a recent result of Rossman and Srinivasan~\cite{RossmanS}.

The key step in the results of~\cite{OW,Amano} is to show that the $\delta$-coin problem can be solved by explicit read-once\footnote{i.e. each variable in the formula appears exactly once} monotone $\AC^0$ formulas of depth $d$ and size $\exp(O(d(1/\delta)^{1/(d-1)})).$ (This is improved to circuits of size $\exp(O((1/\delta)^{1/(d-1)}))$ in~\cite{RossmanS}. However, these circuits are \emph{not} explicit.) Compare this to the trivial solution (i.e. computing Majority on $\Theta(1/\delta^2)$ inputs) that requires $\AC^0$ circuit size $\exp(\Omega((1/\delta)^{2/(d-1)})),$ which is superpolynomially worse.

\paragraph{Known lower bounds.} Lower bounds for the coin problem have also
been subject to a good deal of investigation. Shaltiel and
Viola~\cite{SV} show that if the
$\delta$-coin problem can be solved by a circuit $C$ of size
$s$ and depth $d$, then there is a circuit $C'$ of size
$\poly(s)$ and depth $d+3$ that computes the Majority function on
$n=\Omega(1/\delta)$ inputs. Using H\r{a}stad's lower bound for
Majority~\cite{Hastad}, this implies that any depth-$d$ $\AC^0$ circuit $C$ solving the $\delta$-coin problem must have
size $\exp(\Omega( (1/\delta)^{1/(d+2)} ))).$  Using Razborov and Smolensky's~\cite{Razborov,Smolensky93} lower bounds, this also yields the same lower bound for the more powerful circuit class $\AC^0[\oplus]$.\footnote{Using~\cite{SV} as a black box will yield a weaker lower bound of $\exp(\Omega((1/\delta)^{1/2(d+2)})).$ However, slightly modifying the proof of Shaltiel and Viola, we can obtain circuits of depth $d+3$ that \emph{approximate} the Majority function on $1/\delta^2$ inputs instead of computing the Majority function on $(1/\delta)$ inputs \emph{exactly}. This yields the stronger lower bound given here.} 

In a later result, Aaronson~\cite{Aaronson} observed that a stronger lower bound can be deduced for $\AC^0$ by constructing a circuit $C''$ of depth $d+2$ that only distinguishes inputs of Hamming weight $n/2$ from inputs of Hamming weight $(n/2)-1$. By H\r{a}stad's results, this suffices to recover a lower bound of $\exp(\Omega( (1/\delta)^{1/(d+1)} )))$ for $\AC^0$, but does not imply anything for $\AC^0[\oplus]$ (since the parity function can distinguish between inputs of weight $n/2$ and $(n/2)-1$).

Note that these lower bounds for the $\delta$-coin problem, while exponential, do not meet the upper bounds described above. In fact, they are quasipolynomially weaker.

Lower bounds for the closely related promise and approximate
majorities were proved by Viola~\cite{Viola} and O'Donnell and
Wimmer~\cite{OW} respectively. Viola~\cite{Viola} shows that any $\poly(n)$-sized
depth-$d$ $\AC^0$ circuit cannot compute a promise majority for
$\delta = o(1/(\log n)^{d-2})$. O'Donnell and Wimmer~\cite{OW} show
that any depth-$d$ $\AC^0$ circuit that approximates the Majority
function on 90\% of its inputs must have size
$\exp(\Omega(n^{1/2(d-1)})$. Using the connection between the coin
problem and approximate majority, it follows that any \emph{monotone}
depth-$d$ $\AC^0$ circuit solving the $\delta$-coin problem must have
size $\exp(\Omega((1/\delta)^{1/(d-1)}))$ matching the upper bounds
above. The lower bound of~\cite{OW} is based on the Fourier-analytic
notion of the \emph{Total Influence} of a Boolean function
(see~\cite[Chapter 2]{ODbook}) and standard upper bounds on the total
influence of a Boolean function computed by a small $\AC^0$
circuit~\cite{LMN,Boppana97}.

Using more Fourier analytic ideas~\cite{LMN}, Cohen, Ganor and Raz~\cite{CGR} proved near-optimal $\AC^0$ circuit lower bounds for the $\delta$-coin problem (with no assumptions on monotonicity). They show that any depth-$d$ $\AC^0$ circuit for the $\delta$-coin problem must have size $\exp(\Omega((1/\delta)^{1/(d-1)}))$, nearly matching the upper bound constructions above.

\paragraph{The Coin Problem and Size Hierarchy Theorems for $\AC^0[\oplus]$.} Recall that to prove size-hierarchy theorems for $\AC^0[\oplus],$ we need to come up with explicit functions for which we can prove near-tight lower bounds. One class of functions for which the Razborov-Smolensky proof technique does yield such a lower bound is the class of approximate majorities defined above. Unfortunately, however, this does not yield an explicit separation, since the functions constructed in~\cite{OW,Amano,RossmanS} are randomized and not explicit. These circuits are obtained by starting with explicit large monotone read-once formulas for the coin problem from~\cite{OW,Amano} and replacing each variable with one of the $n$ variables of the approximate majority; one can show that this probabilistic procedure produces an approximate majority with high probability. However, explicitness is then lost.

Our starting point is that instead of working with approximate majorities, we can directly work with the explicit formulas solving the coin problem. As shown in~\cite{OW,Amano}, there are explicit formulas of size $\exp(O(d(1/\delta)^{1/(d-1)}))$ solving the $\delta$-coin problem. Since these yield optimal-sized circuits for computing approximate majorities, it follows that the functions computed by these formulas cannot be computed by much smaller circuits. 

While this is true, nevertheless the explicit formulas of~\cite{OW,Amano} do not yield anything non-trivial by way of size-hierarchy theorems. This is because, as noted above, these formulas are \emph{read-once}. Hence, showing that the underlying functions cannot be computed in smaller size would prove a separation between size $s = O(n)$ circuits and circuits of size much smaller than $n$, which is trivial.

The way we circumvent this obstacle is to construct explicit circuits solving the $\delta$-coin problem with optimal size and much smaller sample complexity. In fact, we are able to bring down the sample complexity from exponential to polynomial in $(1/\delta)$. This allows us to prove a size hierarchy theorem for all $s$ up to $\exp(n^{o(1)})$.

\subsection{Our results for the Coin Problem}

We make progress on the complexity of the coin problem on both the upper bound and lower bound fronts. 

\paragraph{Upper bounds.} Note that the upper bound results known so far
only yield circuit size and depth upper bounds for the coin problem,
and do not say anything about the sample complexity of the
solution. In fact, the explicit $\AC^0$ formulas of O'Donnell and
Wimmer~\cite{OW} and Amano~\cite{Amano} are \emph{read-once} in the
sense that each input corresponds to a distinct input variable. Hence,
these results imply explicit formulas of size
$s = \exp(O(d(1/\delta)^{1/(d-1)}))$ and sample size $\Theta(s)$ for
the $\delta$-coin problem. (Recall that, in contrast, the trivial
solution has sample complexity only $O(1/\delta^2).$) The sample
complexity of these formulas can be reduced to $O(1/\delta^2)$ by a
probabilistic argument (as essentially shown by~\cite{OW,Amano}; more
on this below), but then we no longer have explicit formulas. The
circuit construction of Rossman and Srinivasan~\cite{RossmanS} can be
seen to use $O(1/\delta^2)$ samples, but is again not explicit.

We show that the number of samples can be reduced to $\poly(1/\delta)$ (where the degree of the polynomial depends on the depth $d$ of the circuit), which is the in same ballpark as the trivial solution, while retaining both the size and the explicitness of the formulas. The result is as follows.

\begin{theorem}[Explicit formulas for the coin problem with small sample complexity]
\label{thm:main-intro}
Let $\delta \in (0,1)$ be a parameter and $d\geq 2$ any fixed constant. There is an explicit depth-$d$ monotone $\AC^0$ formula $\Gamma_d$ that solves the $\delta$-coin problem, where $\Gamma_d$ has size $\exp(O(d(1/\delta)^{1/(d-1)}))$ and sample complexity $(1/\delta)^{2^{O(d)}}.$ (All the constants implicit in the $O(\cdot)$ notation are absolute constants.)
\end{theorem}

\noindent
{\bf Approximate majority and the coin problem.} This result may be interpreted as a ``partial derandomization'' of the approximate majority construction of~\cite{OW,Amano} in the following sense. It is implicit in~\cite{OW,Amano} that an approximate majority on $n$ variables can be obtained by starting with a \emph{monotone} circuit $C$ solving the $\delta$-coin problem for $\delta = \Theta(1/\sqrt{n})$, and replacing each input of $C$ with a random input among the $n$ input bits on which we want an approximate majority. While, as noted above, the coin-problem-solving circuits of~\cite{OW,Amano} have exponential sample-complexity, our circuits only have polynomial sample-complexity, leading to a much more randomness-efficient way of constructing such an approximate majority. 

Indeed, this feature of our construction is crucial for proving the Size-Hierarchy theorem for $\AC^0[\oplus]$ circuits (Theorem~\ref{thm:size-hie}).

\paragraph{Lower bounds.} As noted above, Shaltiel and Viola~\cite{SV} 
prove that any $\AC^0[\oplus]$ circuit of depth $d$ solving the $\delta$-coin problem must have size at least $\exp(\Omega((1/\delta)^{1/(d+2)})).$ For the weaker class of $\AC^0$ circuits, Cohen, Ganor and Raz~\cite{CGR} proved an optimal bound of $\exp(\Omega((1/\delta)^{1/(d-1)}))$.
We are also able to strengthen these incomparable results by proving an optimal lower bound for $\AC^0[\oplus]$ circuits. More formally, we prove the following.

\begin{theorem}[Lower bounds for the coin problem]
\label{thm:lb-intro}
Say $g$ is a Boolean function solving the $\delta$-coin problem, then any $\AC^0[\oplus]$ formula of depth $d$ for $g$ must have size $\exp(\Omega(d(1/\delta)^{1/(d-1)}))$. (The $\Omega(\cdot )$ hides an absolute constant.)
\end{theorem}

While the above result is stated for $\AC^0[\oplus]$ \emph{formulas}, it easily implies a $\exp(\Omega((1/\delta)^{1/(d-1)}))$ lower bound for depth-$d$ \emph{circuits}, since any $\AC^0[\oplus]$ circuit of size $s$ and depth $d$ can be converted to an $\AC^0[\oplus]$ formula of size $s^d$ and depth $d$. We thus get a direct extension of the results of Shaltiel and Viola~\cite{SV} and Cohen, Ganor and Raz~\cite{CGR}.

The proof of this result is closely related to the results of Razborov~\cite{Razborov} and Smolensky~\cite{Smolensky} (also see~\cite{szegedy,Smolensky93}) that prove lower bounds for $\AC^0[\oplus]$ circuits computing the Majority function. For \emph{monotone} functions\footnote{Recall that a function $g:\{0,1\}^m\rightarrow \{0,1\}$ is monotone if it is non-decreasing w.r.t. the standard partial order on the hypercube.}, the lower bound immediately follows from the standard lower bounds of ~\cite{Razborov} and~\cite{Smolensky} for approximate majorities\footnote{The standard lower bounds of Razborov and Smolensky are usually stated for computing the hard function (e.g. Majority) \emph{exactly}. However, it is easily seen that the proofs only use the fact that the circuit computes the function on most (say 90\%) of its inputs (see, e.g.~\cite{RossmanS}). In particular, this yields lower bounds even for approximate majorities, which, moreover, turn out to be tight. This can be seen as an alternate proof of the (later) lower bound of O'Donnell and Wimmer~\cite[Theorem 4]{OW} for a stronger class of circuits. (The lower bound of~\cite{OW} only holds for $\AC^0$.)} (actually, we need a slightly stronger lower bound for $\AC^0[\oplus]$ \emph{formulas}) and the reduction~\cite{OW,Amano} from computing approximate majorities to the coin problem outlined above. This special case is already enough for the size-hierarchy theorem stated in Section~\ref{sec:sht}. 

However, to prove the result in the non-monotone setting, it is not clear how to use the lower bounds of~\cite{Razborov,Smolensky} directly. Instead, we use the ideas behind these results, specifically the connections between $\AC^0[\oplus]$ circuits and low-degree polynomials. We show that if a function $g(x_1,\ldots,x_N)$  solves the $\delta$-coin problem, then its degree, as a polynomial from $\F_2[x_1,\ldots,x_N]$, must be at least $\Omega(1/\delta)$ (independent of its sample complexity). From this statement and Razborov's~\cite{Razborov} low-degree polynomial approximations for $\AC^0[\oplus],$ it is easy to infer the lower bound. Further, we think that the statement about polynomials is interesting in its own right.

Note that Theorems~\ref{thm:main-intro} and~\ref{thm:lb-intro} immediately imply the Fixed-depth Size Hierarchy Theorem for $\AC^0[\oplus]$ (Theorem~\ref{thm:size-hie}).

\paragraph{Independent work of Chattopadhyay, Hatami, Lovett and Tal~\cite{CHLT}.} A beautiful recent paper of Chattopadhyay et al.~\cite{CHLT} proves a result on the Fourier spectrum of low-degree polynomials over $\F_2$ (Theorem 3.1 in~\cite{CHLT}) which is equivalent\footnote{We thank Avishay Tal (personal communication) for pointing out to us that the results of~\cite{CHLT} imply the degree lower bounds for the coin problem using an observation of~\cite{CHHL}. This direction in fact works for any class of functions closed under \emph{restrictions} (i.e. setting inputs to constants from $\{0,1\}$).} to the degree lower bound on the coin problem mentioned above. Indeed, the main observation, which is an extension of the Smolensky~\cite{Smolensky,Smolensky93} lower bound for the Majority function, is common to our proof as well as that of~\cite{CHLT}.

\srikanth{Need to add a short paragraph on techniques.}

\subsection{Proof Outline}

\paragraph{Upper bounds.} We start with a description of the read-once formula construction of~\cite{OW, Amano}. In these papers, it is shown that for every $d\geq 2$, there is an explicit read-once formula $F_d$ that solves the $\delta$-coin problem. This formula $F_d$ is defined as follows. We fix a $d\geq 2$ and let $m = \Theta((1/\delta)^{1/(d-1)})$, a large positive integer. We define positive integer parameters $f_1,\ldots,f_d \approx \exp(m),$\footnote{These numbers have to be chosen carefully for the proof, but we do not need to know them exactly here.} and define the formula $F_d$ to be a read-once formula with alternating AND and OR input gates where the gates at height $i$ in the formula all have fan-in $f_i$. (It does not matter if we start with AND gates or OR gates, but for definiteness, let us assume that the bottom layer of gates in the formula is made up of AND gates.) Each leaf of the formula is labelled by a distinct variable (or equivalently, the formula is read-once).  \nutan{Do we really need the construction of $F_d$ here? By simply defining $d,m$ as done here we can state the size of $F_d$ in terms of these.  }\srikanth{Removed.}

The formula $F_d$ is shown to solve the coin problem (for suitable values of $f_1,\ldots,f_d$). Note that the size of the formula, as well as its sample complexity, is $\Theta(f_1\cdots f_d)$, which turns out be $\exp(\Omega(md)) = \exp(\Omega(d(1/\delta)^{1/(d-1)})).$

To show that the formula $F_d$ solves the coin problem, we proceed as follows. For each $i\in \{1,\ldots,d\}$, let us define $\mathrm{Acc}_i^{(0)}$ (resp. $\mathrm{Rej}_i^{(0)}$) to be the probability that some subformula $F_i$ of height $i$ accepts (resp. rejects) an input from the distribution $\mu_0^{N_i}$ (where $N_i$ is the sample complexity of $F_i$). Similarly, also define $\mathrm{Acc}_i^{(1)}$ and $\mathrm{Rej}_i^{(1)}$ w.r.t. the distribution $\mu_1^i.$ Define 
\[
p_i^{(b)} = \min\{\mathrm{Acc}_i^{(b)},\mathrm{Rej}_i^{(b)}\}
\]
for each $b\in \{0,1\}$. Note that these definitions are independent of the exact subformula $F_i$ of height $i$ that we choose. 

It can be shown via a careful analysis that for each odd $i < d$ and each $b\in \{0,1\}$, $p_i^{(b)} = \mathrm{Acc}_i^{(b)} = \Theta(1/2^m)$ (i.e. the acceptance probability is smaller than the rejection probability and is roughly $1/2^m$) and we have
\begin{equation}
\label{eq:ratio}
\frac{p_i^{(1)}}{p_i^{(0)}} = (1+\Theta(m^i\delta)).
\end{equation}
(Note that when $i < d-1,$ the quantity $m^i\delta = o(1)$ and so $p_i^{(0)}$ and $p_i^{(1)}$ are actually quite close to each other.) An analogous fact holds for even $i$ and rejection probabilities, where we now measure $p_i^{(0)}/p_i^{(1)}$ instead. At $i=d-1$, we get that the ratio is in fact a large constant. From here, it is easy to argue that $F_d$ accepts an input from $\mu_1^{N_d}$ w.h.p., and rejects an input from $\mu_0^{N_d}$ w.h.p.. This concludes the proof of the fact that $F_d$ solves the $\delta$-coin problem.

We now describe the ideas behind our derandomization. The precise calculations that are required for the analysis of $F_d$ use crucially the fact that the formulas are read-once. In particular, this implies that we are considering the AND or OR of distinct subformulas, it is easy to compute (using independence) the probability that these formulas accept an input from the distributions $\mu_0^{N_d}$ or $\mu_1^{N_d}.$ In the derandomized formulas that we construct, we can no longer afford read-once formulas, since the size of our formulas is (necessarily) exponential in $(1/\delta),$ but the number of distinct variables (i.e. the sample complexity) is required to be $\poly(1/\delta).$ Thus, we need to be able to carry out the same kinds of precise computations for ANDs or ORs of formulas that share many variables. 

For this, we use a tool from probabilistic combinatorics named \emph{Janson's inequality}~\cite{janson,alon-spencer}. Roughly speaking, this inequality says the following. Say we have a \emph{monotone} formula $F$ over $n$ Boolean variables that is the OR of $M$  subformulas $F_1,\ldots,F_M$, and we want to analyze the probability that $F$ rejects a random input $\bf{x}$ from some product distribution over $\{0,1\}^n.$ Let $p_i$ denote the probability that $F_i$ rejects a random input. If the $F_i$s are variable disjoint, we immediately have that $F$ rejects $\bf{x}$ with probability $\prod_i p_i.$ However, when the $F_i$s are not variable disjoint but \emph{most pairs} of these subformulas are variable disjoint, then Janson's inequality allows us to infer that this probability is \emph{upper bounded} by $\left(\prod_i p_i\right)\cdot (1+\alpha)$ where $\alpha$ is quite small. Furthermore, by the monotonicity of $F$ and the resulting positive correlation between the distinct $F_i$, we immediately see that the probability that $F$ rejects is always \emph{lower bounded} by $\prod_i p_i$ and hence we get
\[
\prod_i p_i \leq \prob{\bf{x}}{\text{$F$ rejects $\bf{x}$}} \leq \left(\prod_i p_i\right)\cdot (1+\alpha).
\]
In other words, the estimate $\prod_i p_i$, which is an exact estimate of the rejection probability of $F$ in the disjoint case, is a good  \emph{multiplicative} approximation to the same quantity in the correlated case. Note that this is \emph{exactly} the kind of approximation that would allow us to recover an inequality of the form in (\ref{eq:ratio}) and allow an analysis similar to that of~\cite{OW,Amano} to go through even in the correlated setting. 

\begin{remark}
\label{rem:janson-intro}
While Janson's inequality has been used earlier in the context of Boolean circuit complexity (for example in the work of Rossman~\cite{Rossman-clique,Rossmanmonclique}), as far as we know, this is the first application in the area of the fact that the inequality actually yields a \emph{multiplicative approximation} to the probability being analyzed. 
\end{remark}

This observation motivates the construction of our derandomized formulas (with only $\poly(1/\delta)$ variables). At each depth $d$, we construct the derandomized formula $\Gamma_d$ as follows. The structure (i.e. fan-ins) of the formula $\Gamma_d$ is exactly the same as that of $F_d$. However, the subformulas of $\Gamma_d$ are not variable disjoint. Instead, we use the $n_d$ variables of $\Gamma_d$ to obtain a family $\mc{F}$ of $f_d$ many sets of size $n_{d-1}$ (one for each subformula of depth $d-1$) in a way that ensures that Janson's inequality can be used to analyze the acceptance or rejection probability of $\Gamma_d$. 

As mentioned above, to apply Janson's inequality, this family $\mc{F}$ must be chosen in a way that ensures that most pairs of sets in $\mc{F}$ are disjoint. It turns out that we also need other properties of this family to ensure that the multiplicative approximation $(1+\alpha)$ is suitably close to $1$. However, we show that standard designs due to Nisan and Wigderson~\cite{Nisan, NW} used in the construction of pseudorandom generators already have these properties (though these properties were not needed in these earlier applications, as far as we know). 

With these properties in place, we can analyze the derandomized formula $\Gamma_d.$ For each subformula $\Gamma$ of depth $i\leq d$, we can define $p_{\Gamma}^{(b)}$ analogously to above. Using a careful analysis, we show that $p_{\Gamma}^{(b)} \in [p_i^{(b)}(1-\alpha_i),p_i^{(b)}(1+\alpha_i)]$ for a suitably small $\alpha_i.$ This allows to infer an analogue of (\ref{eq:ratio}) for $p_\Gamma^{(1)}$ and $p_\Gamma^{(0)}$, which in turn can be used to show (as in the case of $F_d$) that $\Gamma_d$ solves the $\delta$-coin problem. 

\paragraph{Lower bounds.} We now describe the ideas behind the proof of Theorem~\ref{thm:lb-intro}. It follows from the result of O'Donnell and Wimmer~\cite{OW} that there is a close connection between the $\delta$-coin problem and computing an approximate majority on $n$ Boolean inputs. In particular, it follows from this connection that if there is an $\AC^0[\oplus]$ formula $F$ of size $s$ and depth $d$ solving the $\delta$-coin problem for $\delta = \Theta(1/\sqrt{n})$ that \emph{additionally computes a monotone function},\footnote{Note that we are not restricting the formula $F$ itself to be monotone. We only require that it computes a monotone function.} then we also have a formula $F'$ of size $s$ and depth $d$ computing an approximate majority on $n$ inputs. (The formula $F'$ is obtained by substituting each input of $F$ with a uniformly random input among the $n$ inputs to the approximate majority.) Since standard lower bounds for $\AC^0[\oplus]$ formulas~\cite{Razborov,Smolensky,RossmanS} imply lower bounds for computing approximate majorities, we immediately get a lower bound of $\exp(\Omega(d(1/\delta)^{1/(d-1)}))$ for $\AC^0[\oplus]$ formulas $F$ that solve the $\delta$-coin problem \emph{by computing a monotone function.}

For the general case, the above reduction from approximate majorities to the coin problem no longer works and we have to do something different. Our strategy is to look inside the proof of the $\AC^0[\oplus]$ formula lower bounds and use these ideas to prove the general lower bound for the coin problem. In particular, by the polynomial-approximation method due to Razborov~\cite{Razborov} (and a quantitative improvement from~\cite{RossmanS}), it suffices to prove degree lower bounds on polynomials from $\F_2[x_1,\ldots,x_N]$ that solve the $\delta$-coin problem. 

We are able to prove the following theorem in this direction, which we believe is independently interesting.

\begin{theorem}
\label{thm:lbd-polys-intro}
Let $g\in \F_2[x_1,\ldots,x_N]$ solve the $\delta$-coin problem. Then, $\deg(g) = \Omega(1/\delta).$
\end{theorem}

\begin{remark}
\label{rem:lbd-polys}
\begin{enumerate}
\item Note that the degree lower bound in Theorem~\ref{thm:lbd-polys-intro} is independent of the sample complexity $N$ of the underlying function $g$. 
\item The lower bound obtained is tight up to a constant factor. This can be seen by using the fact  that this yields tight lower bounds for the coin problem (which we show), or by directly approximating the Majority function on $1/\delta^2$ bits suitably~\cite{Gopalanetalapprox} to obtain a degree $O(1/\delta)$ polynomial that solves the $\delta$-coin problem.
\item A weaker degree lower bound of $\Omega(1/(\delta\cdot (\log^2 (1/\delta))))$ can be obtained by using an idea of Shaltiel and Viola~\cite{SV}, who show how to use any solution to the coin problem and some additional Boolean circuitry to approximate the Majority function on $1/\delta^2$ inputs. Unfortunately, this weaker degree lower bound only implies a formula lower bound that is superpolynomially weaker than the upper bound.
\item As mentioned above, an independent recent paper of Chattopadhyay et al.~\cite{CHLT} proves a result on the Fourier spectrum of low-degree polynomials which can be used to recover the degree bound  in Theorem~\ref{thm:lbd-polys-intro} (Avishay Tal (personal communication)). Conversely, Theorem~\ref{thm:lbd-polys-intro} can be used to recover the corresponding result of Chattopadhyay et al.~\cite{CHLT}.
\end{enumerate}
\end{remark}

The proof of Theorem~\ref{thm:lbd-polys-intro} is inspired by a standard result in circuit complexity  that says that any polynomial $P\in \F_2[x_1,\ldots,x_n]$ that computes an approximate majority must have degree $\Omega(\sqrt{n}).$ The basic ideas of this proof go back to Smolensky~\cite{Smolensky},\footnote{Though Razborov~\cite{Razborov} was the first to prove an exponential $\AC^0[\oplus]$ circuit lower bound for the Majority function, he did not explicitly prove a lower bound on the degree of approximating polynomials for the Majority function. Instead, he worked with a different symmetric function for the polynomial question.} though the result itself was proved in Szegedy's PhD thesis~\cite{szegedy} and a later paper of Smolensky~\cite{Smolensky93}. Here, we modify a slightly different ``dual'' proof of this result which appears in the work of Kopparty and Srinivasan~\cite{KS}, which itself builds on ideas of Aspnes, Beigel, Furst and Rudich~\cite{ABFR} and Green~\cite{Green}. (The proof idea of Smolensky~\cite{Smolensky} can also be made to work.) 

The first idea is to note that the proof in~\cite{KS} can be modified to prove a lower bound of $\Omega(\sqrt{n})$ on the degree of any $P\in \F_2[x_1,\ldots,x_n]$ that satisfies the following condition: there exist constants $a > b$ such that $P$ agrees with the Majority function on $n$ bits on all but an $\varepsilon$ fraction of inputs of Hamming weight in $[(n/2)-a\sqrt{n},(n/2)-b\sqrt{n}]\cup [(n/2)+b\sqrt{n},(n/2)+a\sqrt{n}]$ (where $\varepsilon$ is suitably small depending on $a,b$).

Using the sampling argument of O'Donnell and Wimmer~\cite{OW} and the above degree lower bound, it follows that if $g$ satisfies the property that it accepts w.h.p. inputs from any product distribution $D_{\alpha}^N$ for $\alpha\in [(1/2)-a\delta,(1/2)-b\delta]$ and rejects w.h.p. inputs from any product distribution $D_{\beta}^N$ for $\beta\in [(1/2)+b\delta,(1/2)+a\delta],$ then the degree of $g$ must be $\Omega(1/\delta).$ \nutan{I am not able to understand this statement?}\srikanth{Fixed.}

But $g$ might not satisfy this hypothesis. Informally, solving the $\delta$-coin problem only means that the acceptance probability of $g$ is small on inputs from $D_{(1-\delta)/2}^N$ and large on inputs from $D_{(1+\delta)/2}^N$. It is not clear that these probabilities will remain small for $\alpha,\beta$ in some intervals of length $\Omega(\delta)$. For example, it may be that the acceptance probability of the polynomial $g$ on distribution $D_{\alpha}^N$ oscillates rapidly for $\alpha\in [(1/2)-a\delta,(1/2)-b\delta]$ even for $a,b$ that are quite close to each other. In this case, however, we observe that $g$ can be used to distinguish $D_{\alpha'}^{N'}$ and $D_{\alpha''}^{N'}$ for $\alpha',\alpha''$ quite close to each other. In other words, we are solving a `harder' coin problem (since $|\alpha'-\alpha''|$ is small). Further, we can show that this new distinguisher, say $g'$, has not much larger degree and sample complexity than the old one. We can thus try to prove the degree lower bound for $g'$ instead. 

We repeat this argument until we can prove a degree lower bound on the new distinguisher $g'$ (which implies a degree lower bound on $g$). We can show that since the sample complexities of successive distinguishers are not increasing too quickly, but the coin problems that they solve are getting much harder, this iteration cannot continue for more than finitely many steps. Hence, after finitely many steps, we will be able to obtain a degree lower bound.

\subsection{Other related work}

The coin problem has also been investigated in other computational models. Brody and Verbin~\cite{BV}, who formally defined the coin problem, studied the complexity of this model in read-once branching programs. Their lower bounds were strengthened by Steinberger~\cite{Steinberger} and Cohen, Ganor and Raz~\cite{CGR}. Lee and Viola~\cite{LV} studied the problem which has also been studied in the model of ``product tests.'' Both these models are incomparable in strength to the constant-depth circuits we study here.

\section{Preliminaries}
\label{sec:prelim}
Throughout this section, let $d\geq 2$ be a fixed constant and $\delta \in (0,1)$ be a parameter. For any $N\geq 1$, let $\mu_0^N$ and $\mu_1^N$ denote the product distributions over $\{0,1\}^N$ where each bit is set to $1$ with probability $(1-\delta)/2$ and $(1+\delta)/2$ respectively.

\subsection{Some technical preliminaries}
\label{sec:tech-prelim}

Throughout, we use $\log (\cdot)$ to denote logarithm to the base $2$ and $\ln (\cdot)$ for the natural logarithm. We use $\exp(x)$ to denote $e^x.$ 

\begin{fact}
\label{fac:exp}
Assume that $x\in [-1/2,1/2].$ Then we have the following chain of inequalities.
\begin{equation}
\label{eq:exp-ineq}
\exp(x - (|x|/2))\mathop{\leq}_{\mathrm{(a)}} \exp(x-x^2) \mathop{\leq}_{\mathrm{(b)}} 1+x\mathop{\leq}_{\mathrm{(c)}} \exp(x)\mathop{\leq}_{\mathrm{(d)}} 1+x+x^2 \mathop{\leq}_{\mathrm{(e)}} 1+x+(|x|/2)
\end{equation}

\end{fact}

The following is an easy consequence of the Chernoff bound. 

\begin{fact}[Error reduction]
\label{fac:err_redn}
Say $g$ solves the coin problem with error $(1/2)-\eta$ for some $\eta > 0$ and let $N$ denote the sample complexity of $g$. Let $G_t:\{0,1\}^{N\cdot t}\rightarrow \{0,1\}$ be defined as follows. On input $x \in \{0,1\}^{N\cdot t},$ 
\[
G_t(x) = \Maj_t(g(x_1,\ldots,x_N),g(x_{N+1},\ldots,x_{2N}),\ldots,g(x_{(t-1)N+1},\ldots,x_{t\cdot N})).
\]
Then, for $t = O(\log(1/\varepsilon)/\eta^2),$ $G_t$ solves the $\delta$-coin problem with error at most $\varepsilon.$
\end{fact}

\subsection{Boolean formulas}
\label{sec:formulas}

We assume standard definitions regarding Boolean circuits. The size of a circuit always refers to the total number of gates (including input gates) in the circuit.

We abuse notation and use \emph{$\AC^0$ formulas of size $s$ and depth $d$} (even for superpolynomial $s$) to denote depth $d$ formulas of size $s$ made up of AND, OR and NOT gates. Similar notation is also used for $\AC^0[\oplus]$ formulas.

Given a Boolean formula $F$, we use $\Vars(F)$ to denote the set of variables that appear as labels of input gates of $F$. 

We say that a Boolean formula family $\{F_n\}_{n\geq 1}$ is explicit if there is a deterministic polynomial-time algorithm which when given as input $n$ (in binary) and the description of two gates $g_1,g_2$ of $F_n$ is able to compute whether there is a wire from $g_1$ to $g_2$ or not. Such a notion of explicitness has been described as \emph{uniformity} in \cite{vollmer}(see Chapter 2 and definition 2.24). \srikanth{Link to reference for this kind of explicitness??}

\subsection{Amano's formula construction}
\label{sec:amano}
In this section we present the construction of a depth $d$ $\AC^0$ formula that solves the $\delta$-coin problem. The construction presented here is by Amano~\cite{Amano}, which works for $d\geq 3$. For $d=2$, a construction was presented by O'Donnell and Wimmer~\cite{OW}. We describe their construction in Section~\ref{sec:ubd-d=2}.

Define $m = \lceil (1/\delta)^{1/(d-1)}\cdot(1/\ln 2)\rceil.$ For $i\in [d-2]$, define $\delta_i$ inductively by $\delta_1 = m\delta$ and $\delta_i = \delta_{i-1}\cdot (m\ln 2).$

Define fan-in parameters $f_1 = m, f_2 = f_3 = \cdots = f_{d-2} = \lceil m\cdot 2^m\cdot \ln 2\rceil, f_{d-1} = C_1\cdot m2^m$ and $f_d = \lceil \exp(C_1\cdot m)\rceil,$ where $C_1= 50.$

Define the formula $F_d$ to be an alternating formula with AND and OR gates such that
\begin{itemize}
\item Each gate at level $i$ above the variables has fan-in $f_i$.
\item The gates at level $1$ (just above the variables) are AND gates.
\item Each leaf is labelled by a distinct variable. 
\end{itemize}

Note that $F_d$ is a formula on $N = \prod_{i\in [d]}f_i  \leq \exp(O(dm))$ variables of size $O(N).$ 

Amano~\cite{Amano} showed that $F_d$ solves the $\delta$-coin problem. We state a more detailed version of his result below. Since this statement does not exactly match the statement in his paper, we give a proof in the appendix.

For each $i\leq d$, let $F_i$ denote any subformula of $F_d$ of depth $i.$ Let $N_i$ denote $|\Vars(F_i)|$ and let $p_i^{(b)}$ denote the probability
\begin{equation}
\label{eq:p-i-def}
p_i^{(b)} = \min\{\prob{\bm{x}\sim \mu_b^{N_i}}{F_i(\bm{x}) = 0}, \prob{\bm{x}\sim \mu_b^{N_i}}{F_i(\bm{x}) = 1}\}.
\end{equation}
Note that the definition of $p_i^{(b)}$ is independent of the exact subformula $F_i$ chosen: any subformula of depth $i$ yields the same value.

\begin{theorem}
\label{thm:amano}
Assume $d\geq 3$ and $F_d$ is defined as above. Then, for small enough $\delta$, we have the following.
\begin{enumerate}
\item For $b,\beta\in \{0,1\}$ and each $i\in [d-1]$ such that $i\equiv \beta \pmod{2}$, we have 
\[
p_i^{(b)} = \prob{\bm{x}\sim \mu_b^{N_i}}{F_i(\bm{x}) = \beta}.
\]
In particular, for any $i\in \{2,\ldots,d-2\}$ and any $b\in \{0,1\}$
\begin{equation}
\label{eq:pivspi-1}
p_i^{(b)} = (1-p_{i-1}^{(b)})^{f_i}.
\end{equation}
\item For $\beta \in \{0,1\}$ and $i\in [d-2]$ such that $i\equiv \beta\pmod{2},$ we have
\begin{align*}
\frac{1}{2^m}(1+\delta_i\exp(-3\delta_i))&\leq  p_i^{(\beta)} \leq \frac{1}{2^m}(1+\delta_i\exp(3\delta_i))\\
\frac{1}{2^m}(1-\delta_i\exp(3\delta_i))&\leq  p_i^{(1-\beta)} \leq \frac{1}{2^m}(1-\delta_i\exp(-3\delta_i))
\end{align*}
\item Say $d-1\equiv \beta \pmod{2}.$ Then
\[
p_{d-1}^{(\beta)} \geq \exp(-C_1m + C_2) \text{ and } p_{d-1}^{(1-\beta)} \leq \exp(-C_1m - C_2)
\]
where $C_2 = C_1/10.$
\item For each $b\in \{0,1\}$, $\prob{\bm{x}\sim \mu^N_b}{F_d(\bm{x}) = 1-b}\leq 0.05.$ In particular, $F_d$ solves the $\delta$-coin problem.
\end{enumerate}
\end{theorem}

\begin{observation}
\label{obs:pifi}
For any $i \in \{2, \ldots, d\}$ and $b \in \{0,1\}$, $p_{i-1}^{(b)} \cdot f_i \leq 50m$. 
\end{observation}

\begin{remark}
\label{rem:RST}
A similar construction to Amano's formula above was used by Rossman, Servedio and Tan~\cite{RST} to prove an average-case \emph{Depth-hierarchy} theorem for $\AC^0$ circuits. Their construction was motivated by the \emph{Sipser functions} used in the work of Sipser~\cite{SipserDH} and H\r{a}stad~\cite{Hastad} to prove worst-case Depth-hierarchy theorems. 
\end{remark}

\subsection{Janson's inequality}
\label{sec:janson}

We state Janson's inequality~\cite{janson} in the language of Boolean circuits. The standard proof due to Boppana and Spencer (see, e.g.~\cite[Chapter 8]{alon-spencer}) easily yields this statement. Since Janson's inequality is not normally presented in this language, we include a proof in the appendix for completeness. 

\begin{theorem}[Janson's inequality]
\label{thm:janson}

Let $C_1,\ldots,C_M$ be any monotone Boolean circuits over inputs $x_1,\ldots,x_N,$ and let $C$ denote $\bigvee_{i\in [M]}C_i.$ For each distinct $i,j\in [M]$, we use $i \sim j$ to denote the fact that $\Vars(C_i)\cap \Vars(C_j)\neq \emptyset$.  Assume each $\bm{x}_j$ ($j\in [n]$) is chosen independently to be $1$ with probability $p_i\in [0,1]$, and that under this distribution, we have $\max_{i\in [M]}\prob{\bm{x}}{C_i(\bm{x}) = 1}\leq 1/2.$ Then, we have
\begin{equation}
\label{eq:janson}
\prod_{i\in [M]}\prob{\bm{x}}{C_i(\bm{x}) = 0} \leq \prob{\bm{x}}{C(\bm{x}) = 0} \leq \left(\prod_{i\in [M]}\prob{\bm{x}}{C_i(\bm{x}) = 0}\right) \cdot \exp(2\Delta)
\end{equation}
where $\Delta := \sum_{i < j: i\sim j} \prob{\bm{x}}{(C_i(\bm{x})=1) \wedge (C_j(\bm{x}) = 1)}.$
\end{theorem}

\begin{remark}
\label{rem:janson}
By using DeMorgan's law, a similar statement also holds for the probability that the conjunction $C' = \bigwedge_{i\in [M]}C_i$ takes the value $1$. More precisely, if $\max_{i\in [M]}\prob{\bm{x}}{C_i(\bm{x}) = 0}\leq 1/2$, we have
\begin{equation}
\label{eq:janson-and}
\prod_{i\in [M]}\prob{\bm{x}}{C_i(\bm{x}) = 1} \leq \prob{\bm{x}}{C'(\bm{x}) = 1} \leq \left(\prod_{i\in [M]}\prob{\bm{x}}{C_i(\bm{x}) = 1}\right) \cdot \exp(2\Delta)
\end{equation}
where $\Delta := \sum_{i < j: i\sim j} \prob{\bm{x}}{(C_i(\bm{x})=0) \wedge (C_j(\bm{x}) = 0)}.$
\end{remark}
\srikanth{move to appendix and add pointer.}

\section{Design construction}
\label{sec:design}

In order to define a derandomized version of the formulas in Section~\ref{sec:amano}, we will need a suitable notion of a combinatorial design. The following definition of a combinatorial design refines the well-known notion of a Nisan-Wigderson design from the work of \cite{Nisan,NW}. We give a construction of our combinatorial design by  using a construction of Nisan-Wigderson design from~\cite{NW} and showing that this construction in fact satisfies the additional properties we need. 

\begin{definition}[Combinatorial Designs]
\label{def:design}
For positive integers $N_1,N_2,M,\ell$ and $\gamma,\eta\in (0,1),$ an \emph{$(N_1,M,N_2,\ell,\gamma,\eta)$-Combinatorial Design} is a family $\mc{F}$ of subsets of $[N_1]$ such that
\begin{enumerate}
\item $|\mc{F}| \geq M,$
\item $\mc{F}\subseteq \binom{[N_1]}{N_2}$ (i.e. every set in $\mc{F}$ has size $N_2$),
\item Given any distinct $S,T\in \mc{F}$ we have $|S\cap T|\leq \ell,$
\item For any $a\in [N_2]$, we have $|\{S\in \mc{F}\ |\ S\ni a\}|\leq \gamma\cdot M,$
\item For any $i\in [\ell]$, we have $|\{\{S,T\}\subseteq \mc{F}\ |\ S\neq T, |S\cap T| = i\}| \leq \eta^i\cdot M^2.$
\end{enumerate}
\end{definition}

The main result of this section is the following.

\begin{lemma}[Construction of Combinatorial design]
\label{lem:design}
Given positive integers $N_2$ and $M$ and real parameters $\gamma,\eta\in (0,1)$ satisfying $N_2\geq (\log M)/10$, $M\geq 10\cdot N_2/\eta,$ and $\gamma\geq \eta/N_2,$ there exist positive integers $\ell= \Theta \left( \log M/\log(N_2/\eta)\right)$ and $Q = O((N_2/\eta)^{1+1/\ell})$ and an $(N_1 = Q\cdot N_2,M,N_2,\ell,\gamma,\eta)$-combinatorial design.

Further, the design is explicit in the following sense. Identify $[N_1]$ with $[N_2]\times [Q]$ via the bijection $\rho:[N_1]\rightarrow [N_2]\times [Q]$ such that $\rho(i) = (j,k)$ where $i = (k-1)N_2 + j$. Then, each set in $\mc{F}$ is of the form $\{(1,k_1),\ldots,(N_2,k_{N_2})\}$ for some $k_1,\ldots,k_{N_2}\in [Q].$ Finally, there is a deterministic algorithm $\mc{A}$, which when given as input an $i\in [|\mc{F}|]$ and a $j\in [N_2]$, produces $k_j\in [Q]$ in $\poly(\log M)$ time.
\end{lemma}

\begin{proof}
\newcommand{\vecb}{\mathbf{b}}

Define $\ell$ to be the largest integer such that $M^{1/\ell} \geq 10\cdot N_2/\eta$: note that $\ell \geq 1$ by our assumption that $M\geq 10\cdot N_2/\eta.$ Thus, we have
\begin{equation}
\label{eq:des-ell}
\ell \leq \frac{\log M}{\log(10\cdot N_2/\eta)}\leq \frac{\log M}{\log \log M}
\end{equation}
and also
\begin{equation}
\label{eq:des-ell+1}
M^{1/(\ell+1)} < \frac{10\cdot N_2}{\eta}.
\end{equation}
Define the parameter $Q_1 = \lceil M^{1/\ell}\rceil.$ We have
\begin{equation}
\label{eq:des-Q-1}
\frac{10\cdot N_2}{\eta}\leq M^{1/\ell} \leq Q_1 \leq 2M^{1/\ell}\leq O\left(\left(\frac{N_2}{\eta}\right)^{1+\frac{1}{\ell}}\right)
\end{equation}
where we used (\ref{eq:des-ell+1}) for the last inequality. 

Let $Q$ be the smallest power of $2$ greater than or equal to $Q_1$ and let $\F_Q$ be a finite field of size $Q.$ By a result of Shoup~\cite{shoup}, we can construct in time $\poly(\log Q) = \poly(\log M)$ time an implicit representation of $\F_Q$ where each element of $\F_Q$ is identified with an element of $\{0,1\}^{\log Q}$ and arithmetic can be performed in time $\poly(\log Q)$. Fix such a representation of $\F_Q.$

Let $A\subseteq \F_Q$ be any subset of size $N_2$ (note that by (\ref{eq:des-Q-1}) we have $N_2\leq Q_1$ which is at most $Q$) and let $B\subseteq \F_Q$ be any fixed subset of size $Q_1$. Let $A_1\subseteq A$ be a set of size $\ell$ (note that by (\ref{eq:des-ell}) $\ell\leq (\log M)/10$ which is at most $N_2$ by assumption).  

Fix $N_1 = Q\cdot N_2$ and identify $[N_1]$ with the set $A\times \F_Q$ in an arbitrary way. Assume that $A = \{a_1,\ldots,a_{N_2}\}$ and $A_1 = \{a_1,\ldots,a_{\ell}\}.$ We define $\mc{P}$ to be the set of all polynomials $P\in \F_Q[x]$ of degree at most $\ell-1$ such that $P(a)\in B$ for each $a\in A_1.$

We are now ready to define the family $\mc{F}.$ For each $\vecb = (b_1,\ldots,b_\ell)\in A^\ell$, we define the polynomial $P_\vecb(x)$ to be the unique polynomial in $\mc{P}$ such that $P(a_i) = b_i$ for each $i\in [\ell]$ (note that $P$ is uniquely defined since any polynomial of degree at most $\ell-1$ can be specified by its evaluations at any $\ell$ distinct points). We add the set $S_{\vecb}\subseteq A\times \F_Q$ to $\mc{F},$ where $S_{\vecb}$ is defined by
\begin{equation}
\label{eq:des-S-b}
S_\vecb = \{(a_i,P_\vecb(a_i))\ |\ i\in [N_2]\}.
\end{equation}
In words, $S_{\vecb}$ is the graph of the polynomial $P_\vecb$ restricted to the domain $A.$ 

We now show that $\mc{F}$ is indeed a $(N_1,M,N_2,\ell,\gamma,\eta)$-combinatorial design.
\begin{enumerate}
\item For distinct $\vecb,\vecb'\in B^\ell$, the sets $|S_\vecb\cap S_{\vecb'}|\leq \ell$ since the graphs of the distinct polynomials $P_\vecb$ and $P_{\vecb'}$ can intersect at at most $\ell-1$ points. In particular, we have $|\mc{F}| = |B|^\ell = Q_1^\ell \geq M$ (by (\ref{eq:des-Q-1})). Further, we also have that any pair of distinct sets in $\mc{F}$ have an intersection of at most $\ell.$ This proves properties 1 and 3 in Definition~\ref{def:design} above.
\item Each set in $\mc{F}$ has size $N_2$, since it is of the form $\{(a_i,b_i)\ |\ i\in [N_2]\}$ for some choice of $b_1,\ldots,b_{N_2}\in \F_Q.$ This proves property 2.
\item We now consider property 4. Fix any $(a,b) \in A\times \F_Q.$  If $(a,b)\in S\in \mc{F},$ then $S$ is the graph of a polynomial $P\in \mc{P}$ such that $P(a) = b$. To uniquely specify such a polynomial, it suffices to provide its evaluations at any $\ell-1$ other points. We choose the evaluation points to be a fixed set $A_1'\subseteq A_1\setminus\{a\}$ of size $\ell-1.$ Since $P(a')\in B$ for each $a'\in A_1',$ there are at most $|B|^{\ell-1} = Q_1^{\ell-1}$ many choices for these evaluations, which yields the same bound for the number of sets $S\in \mc{F}$ such that $(a,b)\in S$. 

Hence, we have 
\[
|\{S\in \mc{F}\ |\ S\ni (a,b)\}| \leq \frac{Q_1^{\ell}}{Q_1} = \frac{\left(\lceil M^{1/\ell}\rceil\right)^{\ell}}{Q_1}\leq  \frac{\eta}{N_2}\cdot M \leq \gamma\cdot M
\]
where the final inequality follows from our assumption that $\gamma\geq \eta/N_2$, and the second last inequality uses the fact that $Q_1\geq 10\cdot N_2/\eta$ and
\begin{equation}
\label{eq:des-ceil}
\left(\lceil M^{1/\ell}\rceil\right)^{\ell} \leq \left( M^{1/\ell} + 1\right)^{\ell} = M\cdot \left(1+\frac{1}{M^{1/\ell}}\right)^\ell \leq M\cdot \left(1+\frac{1}{\ell}\right)^\ell \leq 3M
\end{equation}
(using $\ell^\ell\leq M$, which follows from (\ref{eq:des-ell}), for the second-last inequality).

\item For property 5, we use a similar argument to property 4. Fix distinct sets $S,T\in \mc{F}$ such that $|S\cap T| = i.$ The sets $S$ and $T$ are graphs of distinct polynomials $P_1,P_2\in \mc{P}$ respectively that agree in $i$ places. We bound the number of such pairs of polynomials. 

The number of choices for $S$, and hence $P_1$, is exactly $|\mc{F}| = Q_1^\ell.$ Given $P_1$, we can specify $P_2$ as follows. 
\begin{itemize}
\item Specify a set $A'\subseteq A$ of size $i$ such that $P_1$ and $P_2$ agree on $A'.$ This gives the evaluation of $P_2$ at $i$ points. Further, the number of such $A'$ is $\binom{N_2}{i}\leq N_2^i.$
\item Specify the evaluation of $P_2$ at the first $\ell-i$ points from $A_1\setminus A'.$ This gives the evaluation of $P_2$ at $\ell-i$ points outside $A'$ and hence specifies $P_2$ exactly. The number of possible evaluations is $|B|^{\ell-i} = Q_1^{\ell-i.}$ 
\end{itemize}
Hence, the number of pairs of polynomials $(P_1,P_2)$ whose graphs agree at $i$ points is at most 
\[
Q_1^\ell\cdot N_2^i \cdot Q_1^{\ell-i} = Q_1^{2\ell}\cdot \left(\frac{N_2}{Q_1}\right)^i = \left(\lceil M^{1/\ell}\rceil\right)^{2\ell}\cdot \left(\frac{N_2}{Q_1}\right)^i\leq 9M^2\cdot \left(\frac{N_2}{Q_1}\right)^i \leq 9M^2\cdot \left(\frac{\eta}{10}\right)^i\leq \eta^i\cdot M^2
\]
where for the first inequality we have used (\ref{eq:des-ceil}) and the second inequality follows from the fact that $Q_1\geq 10\cdot N_2/\eta.$
\end{enumerate}

We have thus shown that $\mc{F}$ is indeed a $(N_1,M,N_2,\ell,\gamma,\eta)$-combinatorial design as required. 

The explicitness of the design follows easily from its definition.
\end{proof}

\section{Proof of Theorem~\ref{thm:main-intro} for $d=2$}
\label{sec:ubd-d=2}
In this section we will present the proof of Theorem~\ref{thm:main-intro} for the special case when $d=2$. The proof is quite similar to the case for general $d$, but is somewhat simpler (as the construction of the $\AC^0$ formulas is simpler) and illustrates many of the ideas of the general proof.

Throughout this section, let $\delta$ be a parameter going to $0$. 

We start by stating a result of O'Donnell and Wimmer~\cite{OW}, who gave a depth-$2$ $\AC^0$ formula for solving the $\delta$-coin problem. Formally, they defined a depth-$2$ circuit as follows. 

Let $C_0 \geq 10$. Let $m = 1/\delta\cdot C_0$, $f_1=m$, $f_2=2^m$, where $C=2^{C_0}$. The formula $F_2$ is defined as follows:
\begin{itemize}
\item At layer $1$ we have AND gates and the fan-in of each AND gate is $f_1$. 
\item At layer $2$ we have a single OR gate with fan-in $f_2$. 
\item Each leaf is labelled with distinct variables. 
\end{itemize}

For $F_2$ defined as above~\cite{OW} proved the following theorem.
\begin{theorem}[\cite{OW}]
\label{thm:ow}
Let $N =f_1\cdot f_2$. For each $b \in \{0,1\}$ $$\prob{\bm{x} \sim \mu_b^N}{F_2(\bm{x}) = 1-b} \leq 0.05,$$ i.e. specifically $F_2$ solves the $\delta$-coin problem. 
\end{theorem}
Here, the number of inputs is $N$ and the size of $F_2$ is also $O(N)$. We now give a construction of an explicit depth $2$ formula of the same size as in the theorem above which solves the $\delta$-coin problem, but using far fewer inputs.  We achieve this by an application of the Janson's inequality coupled with our combinatorial design. 

\srikanth{Need to hyphenate throughout. "depth 2 formula" should be "depth-2 formula" throughout.}

We now describe the construction of such a depth $2$ formula $\Gamma_2$. Fix $m,f_1,f_2$ as above. Define parameters $\gamma=1$, $\eta= 1/(16\cdot (\frac{1+\delta}{2})^{m} \cdot f_2) = {1}/({16\cdot (1+\delta)^m})$.  Let $\mc{F}$ be an $(n, f_2, f_1, \ell, \gamma, \eta)$-design  obtained using Lemma~\ref{lem:design}. We are now ready to define $\Gamma_2$. 
\begin{itemize}
\item Let $S_1, S_2, \ldots, S_{f_2} \in \binom{[n]}{f_1}$ be the first $f_2$ sets in the $(n, f_2, f_1, \ell, \gamma, \eta)$-design $\mc{F}$. At layer $1$ we have $f_2$ many AND gates, say $\Gamma_1^1, \ldots, \Gamma_1^{f_2}$, with fan-in $f_1$ each. For each $i \in [f_2]$, the inputs of the gate $\Gamma_1^i$ are the variables indexed by the set $S_i$. 
\item At layer $2$ we have a single gate, which is an OR of $\Gamma_1^1, \ldots, \Gamma_1^{f_2}$. 
\end{itemize}

With this definition of $\Gamma_2$, we now prove Theorem~\ref{thm:main-intro}. From the definition of the parameters, it can be checked that $\eta=\Theta(1)$. Therefore, we get $\ell = \Theta(m/\log m)$ and $Q=O(f_1/\eta)^{1+1/\ell} = O((1/\delta))$.  Therefore, the number of inputs in the formula is $N = O(Q \cdot f_1) = O(1/\delta^2)$ and the size of the formula is $O(f_1\cdot f_2) = \exp(O(1/\delta))$. 

The only thing we need to prove now is that
for any $b \in \{0,1\}$, $\prob{\bm{x}\sim \mu_b^{n}}{\Gamma_2(\bm{x}) = 1-b}\leq 0.1$. 
Let $q^{(0)} = (1-\delta)/2$ and $q^{(1)} = (1+\delta)/2$. Let $p_1^{(0)} = \left(\frac{1-\delta}{2}\right)^m$ and $p_1^{(1)} = \left(\frac{1+\delta}{2}\right)^m$. Note that $p_1^{(b)}$ ($b\in \{0,1\}$) is the probability that each subformula $\Gamma_1^i$ accepts on a random input $\bm{x}$ chosen from the distribution $\mu_b^n.$

Let $b=0$. In this case 
\begin{align*}
\prob{\bm{x}\sim \mu_0^n}{\Gamma_2(\bm{x}) = 1} & = \prob{\bm{x}\sim \mu_0^n}{\exists i \in [f_2]: \Gamma_1^i(\bm{x}) = 1} \\
& \leq f_2 \cdot p_1^{(0)} = (1-\delta)^m \leq \exp(-C_0) \leq 0.1.
\end{align*}

Here the first inequality is due to a union bound. The other inequalities are obtained by simple substitutions of the parameters and using (\ref{eq:exp-ineq}). 

Now consider the $b=1$ case. Here $\prob{\bm{x}\sim \mu_1^n}{\Gamma_2(\bm{x}) = 0}  = \prob{\bm{x}\sim \mu_1^n}{\forall i \in [f_2]: \Gamma_1^i(\bm{x}) = 0}$. Now we would like to bound this using  Janson's inequality (Theorem~\ref{thm:janson}).
Applying Janson's inequality, we get

\begin{align}
\prob{\bm{x}\sim \mu_1^n}{\Gamma_2(\bm{x}) = 0}  & = \prob{\bm{x}\sim \mu_1^n}{\forall i \in [f_2]: \Gamma_1^i(\bm{x}) = 0}\nonumber \\
& \leq \prod_{i \in [f_2]} \prob{\bm{x}\sim \mu_1^n}{\Gamma_1^i(\bm{x}) = 0} \cdot \exp(2\Delta)\nonumber\\
& \leq (1-p_1^{(1)})^{f_2} \cdot \exp(2\Delta) \nonumber \\
& \leq \exp(-p_1^{(1)}f_2 + 2\Delta), \label{eq:janson-d=2}
\end{align}
where 
\[\Delta  = \sum_{\substack{j < k:\\ \Vars(\Gamma_1^j)\cap \Vars(\Gamma_1^k)\neq \emptyset}} \prob{\bm{x}\sim \mu_1^n}{(\Gamma_1^j(\bm{x}) =1) \wedge (\Gamma_1^k(\bm{x}) =1)}.\]
 We will now obtain a bound on $\Delta$. 
\begin{align*}
\Delta & = \sum_{\substack{j < k:\\ \Vars(\Gamma_1^j)\cap \Vars(\Gamma_1^k)\neq \emptyset}} \prob{\bm{x}\sim \mu_1^n}{(\Gamma_1^j(\bm{x}) =1) \wedge (\Gamma_1^k(\bm{x}) =1)}. \\
& = \sum_{r=1}^\ell \sum_{\substack{j < k:\\ |\Vars(\Gamma_1^j)\cap \Vars(\Gamma_1^k)| = r}} \prob{\bm{x}\sim \mu_1^n}{(\Gamma_1^j(\bm{x}) =1) \wedge (\Gamma_1^k(\bm{x}) =1)}
\end{align*}

As $\Gamma_1^j$ and $\Gamma_1^k$ are both ANDs of size $m$, $\Gamma_1^j \wedge \Gamma_1^k$ is an AND of size $(2m-|\Vars(\Gamma_1^j)\cap \Vars(\Gamma_1^k)|)$. Therefore, we get 

\begin{align*}
\Delta& = \sum_{r=1}^\ell \sum_{\substack{j < k:\\ |\Vars(\Gamma_1^j)\cap \Vars(\Gamma_1^k)| = r}} \prob{\bm{x}\sim \mu_1^n}{((\Gamma_1^j \wedge \Gamma_1^k)(\bm{x}) =1)} \\
& = \sum_{r=1}^\ell \sum_{\substack{j < k:\\ |\Vars(\Gamma_1^j)\cap \Vars(\Gamma_1^k)| = r}} \left(\frac{1+\delta}{2}\right)^{2m-r} \\
& = \sum_{r=1}^\ell \sum_{\substack{j < k:\\ |\Vars(\Gamma_1^j)\cap \Vars(\Gamma_1^k)| = r}} \frac{(p_1^{(1)})^2}{((1+\delta)/2)^r} \\
& = \sum_{r=1}^\ell \frac{(p_1^{(1)})^2}{(q^{(1)})^r} \cdot |\{(j,k) \mid j < k \text{ and } |\Vars(\Gamma_1^j)\cap \Vars(\Gamma_1^k)| = r\}| \\
\end{align*}

From the construction of the formula and the combinatorial design $\mc{F}$, we know that $|\{(j,k) \mid j < k \text{ and } |\Vars(\Gamma_1^j)\cap \Vars(\Gamma_1^k)| = r\}| \leq \eta^r f_2^2$. We can also bound $1/q^{(1)}$ by a small constant, say $3$. 

Therefore, we can simplify the above equation as follows:
\begin{align}
\Delta & \leq \sum_{r=1}^\ell {(p_1^{(1)})^2} \cdot 3^r  \cdot \eta^r f_2^2 \nonumber \\
& = (p_1^{(1)})^2 \cdot f_2^2 \sum_{r=1}^\ell 3^r  \cdot \eta^r \nonumber \\
& \leq (p_1^{(1)})^2 \cdot f_2^2 \cdot 4 \cdot \eta \label{eq:delta-bound}
\end{align}
using the fact that $3\eta \leq 1/4$ as $\eta \leq 1/16.$

Now, by using our setting of $\eta = 1/(16 \cdot p_1^{(1)}\cdot f_2)$ in (\ref{eq:delta-bound}), we get $\Delta \leq p_1^{(1)}f_2/4$. Using this value of $\Delta$ in (\ref{eq:janson-d=2}), we get $\prob{\bm{x}\sim \mu_1^n}{\Gamma_2(\bm{x}) = 0}  \leq \exp(-\frac{p_1^{(1)}\cdot f_2}{2}) \leq 0.1$, by our choice of parameters. 
This completes the proof of Theorem~\ref{thm:main-intro} for $d=2$.

\section{Proof of Theorem~\ref{thm:main-intro} for $d\geq 3$}
\label{sec:ubd}

Throughout this section, fix a constant depth $d\geq 3$ and a parameter $\delta\in (0,1).$ The parameter $\delta$ is assumed to be asymptotically converging to $0$. 

We also assume the notation from Section~\ref{sec:amano}. 

\subsection{Definition of the formula $\Gamma_d$}
\label{sec:gamma-d}

The formula $\Gamma_d$ is an alternating monotone depth-$d$ formula made up of AND and OR gates. The structure of the formula and the labels of the gates are the same as in the formula $F_d$ defined in Section~\ref{sec:amano}. However, the leaves are labelled with only $\poly(m)$ distinct variables.

We now proceed to the formal definition. We iteratively define a sequence of formulas $\Gamma_1,\ldots,\Gamma_d$ (where $\Gamma_i$ has depth $i$) as follows. Define the parameters $\gamma$ and $\eta$ by 
\begin{equation}
\label{eq:def-gam-eta}
\gamma = \frac{1}{m^3} \text{  and   } \eta = \frac{1}{m^{10d}}.
\end{equation}

\begin{itemize}
\item $\Gamma_1$ is just an AND of $n_1=m$ distinct variables. 
\item Recall that for $i\geq 2$, any gate at level $i$ in the formula $F_d$ has fan-in $f_i$ for $f_i = \exp(\Theta(m)).$ For each $i\in \{2,\ldots,d\}$, define $n_i$ so that by Lemma~\ref{lem:design}, we have an explicit $(n_i,f_i,n_{i-1},\ell,\gamma,\eta)$-combinatorial design $\mc{F}_i$ where $\ell = \Theta(\log f_i/\log(n_{i-1}/\eta)).$ Note that $n_i = O((n_{i-1})^{2+1/\ell}/\eta^{1+1/\ell}) \leq n_{i-1}^3/\eta^2$.

The formula $\Gamma_i$ is defined on a set $X$ of $n_i$ variables by taking the OR/AND (depending on whether $i$ is even or odd respectively) of $f_i$ copies of $\Gamma_{i-1}$, each defined on a distinct subset $Y\subseteq X$ of $n_{i-1}$ variables obtained from the combinatorial design $\mc{F}_i.$ 

Formally, let $S_1,\ldots,S_{f_i}\in \binom{[n_i]}{n_{i-1}}$ be the first $f_i$ many sets in the design $\mc{F}_i$ (in lexicographic order, say). Identifying $[n_i]$ with the variable set $X$ of $\Gamma_i,$ we obtain corresponding subsets $Y_1,\ldots,Y_{f_i}$ of $X.$ The formula $\Gamma_i$ is an OR/AND of $f_i$ many subformulas $\Gamma_i^1,\ldots,\Gamma_i^{f_i}$ where the $j$th subformula $\Gamma_i^j$ is a copy of $\Gamma_{i-1}$ with variable set $Y_j.$
\end{itemize}

\begin{observation}
\label{obs:gamma-d}
The size of $\Gamma_d$ is $\exp(O(dm)).$ The number of variables appearing in $\Gamma_d$ is $n_d = m^{2^{O(d)}}.$
\end{observation}

\paragraph{Explicitness of the formula $\Gamma_d.$} The structure of the formula is determined completely by the parameter $\delta.$ Thus to argue that the formula $\Gamma_d$ is explicit, it suffices to show that the labels of the input gates can be computed efficiently. Note that the inputs are in $1$-$1$ correspondence with the set $[f_d]\times [f_{d-1}]\times \cdots [f_2]\times [f_1].$

Let $\Gamma_i$ be any subformula of $\Gamma_d$ of depth $i.$ If $i=1$, then $\Gamma_i$ is simply an AND of $m=f_1$ variables and we identify its variable set with $[f_1]$. When $i > 1$, by the properties of the design constructed in Lemma~\ref{lem:design}, we see that the set $\Vars(\Gamma_i)$ is in a natural $1$-$1$ correspondence with the set $\Vars(\Gamma_{i-1})\times [Q_i]$ where $\Gamma_{i-1}$ is any subformula of depth $i-1$ and $Q_i = n_i/n_{i-1}$. Each subformula $\Gamma_i^j$ ($j\in [f_i]$) of depth $i-1$ in $\Gamma_i$ has as its variable set a set of the form $\{(x,k_x)\ |\ x\in \Vars(\Gamma_{i-1}), k_x\in [Q_i]\}.$

Further, by the explicitness properties of the design constructed in Lemma~\ref{lem:design}, we see that given any $x\in \Vars(\Gamma_{i-1})$ and $j\in [f_i]$, we can find in $\poly(\log(f_i)) \leq \poly(m)$ time the variable $(x,k)\in \Vars(\Gamma_i)$ that belongs to $\Vars(\Gamma_i^j).$ Equivalently, given a leaf $\ell = (j_i,\ldots,j_1)\in [f_i]\times \cdots\times [f_1]$ of $\Gamma_i$ and the variable $x\in \Vars(\Gamma_{i-1})$ corresponding to the leaf $(j_{i-1},\ldots,j_1)$ in $\Gamma_{i-1},$ we can find the variable labelling $\ell$ in $\poly(m)$ time. Using this algorithm and a recursive procedure to find the variable $x$, we see that the variable labelling the leaf $\ell$ can be found in $\poly(m)$ time. In particular, given a leaf of $\Gamma_d,$ the variable labelling it can be found in $\poly(m)$ time.

Thus, the formula $\Gamma_d$ is explicit.


\subsection{Analysis of $\Gamma_d$}
\label{sec:analysis-gamma-d}

In this section, we will show that $\Gamma_d$ distinguishes between the distributions $\mu_0^{n_d}$ and $\mu_1^{n_d}$ as defined in Definition~\ref{def:coinproblem}. For brevity, we use $n$ to denote $n_d$.  

Fix any subformula $\Gamma$ of $\Gamma_d$ and $b\in \{0,1\}$. Assume $\Gamma$ has depth $i\in [d]$ and $\beta \in \{0,1\}$ is such that $i\equiv \beta\pmod{2}.$ We define $p_{\Gamma}^{(b)} = \prob{\bm{x}\sim \mu_b^n}{\Gamma(\bm{x}) = \beta}.$ Assume that $\Gamma$ is an OR/AND of depth-$(i-1)$ subformulas $\Gamma^1,\ldots,\Gamma^f.$ We define
\begin{equation}
\label{eq:delta-gamma}
\Delta_{\Gamma}^{(b)} = \sum_{\substack{j < k:\\ \Vars(\Gamma^j)\cap\\ \Vars(\Gamma^k)\neq \emptyset}} \prob{\bm{x}\sim \mu_b^n}{(\Gamma^j(\bm{x}) =1-\beta) \wedge (\Gamma^k(\bm{x}) =1-\beta)}. 
\end{equation}

The following lemma is the main technical lemma of this section. Along with Theorem~\ref{thm:amano}, it easily implies Theorem~\ref{thm:main-intro} (as we show below).

\begin{lemma}
\label{lem:main}
Let $\Gamma_d$ be as constructed above. Then for each $i\in \{2,\ldots,d\}$, each $b\in \{0,1\}$, and any subformula $\Gamma$ of depth $i$, we have the following. 
\begin{enumerate}
\item $p_{\Gamma}^{(b)} \in [p_i^{(b)}(1-\eta\cdot (C_3m)^i), p_i^{(b)}(1+\eta\cdot (C_3m)^i)]$ where $C_3=1000.$
\item $\Delta_{\Gamma}^{(b)} \leq (C_4m)^2\cdot \eta$ where $C_4=100.$
\end{enumerate}
\end{lemma}

Assuming the above lemma, we first prove Theorem~\ref{thm:main-intro}.

\begin{proof}[Proof of Theorem~\ref{thm:main-intro}]
We use the explicit formula $\Gamma_d$ described above. By Lemma~\ref{lem:main} applied in the case that $i=d$, it follows that for each $b\in \{0,1\}$
\[
|\prob{\bm{x}\sim \mu_b^n}{\Gamma_d(\bm{x}) = 1-b} - \prob{\bm{x}\sim \mu_b^n}{F_d(\bm{x}) = 1-b}| = |p_{\Gamma_d}^{(b)}-p_d^{(b)}| \leq p^{(b)}_d \cdot \eta (C_3 m)^d = o(1).
\]
In particular, using Theorem~\ref{thm:amano}, it follows that $\prob{\bm{x}\sim \mu_b^n}{\Gamma_d(\bm{x}) = 1-b}\leq 0.1$ and hence $\Gamma_d$ solves the $\delta$-coin problem. The sample complexity of $\Gamma_d$ is $m^{2^{O(d)}} = (1/\delta)^{2^{O(d)}}$ by construction.
\end{proof}


\begin{proof}[Proof of Lemma~\ref{lem:main}]
We prove the lemma by induction on $i$. The base case is when $i=2$. This proof is quite similar to the proof of the $d=2$ case from Section~\ref{sec:ubd-d=2}.

\paragraph{Base case, i.e. $i=2$:} Recall that for $i=2$, $\Gamma$ is an OR of $f_2$-many subformulas $\Gamma^1, \Gamma^2, \ldots, \Gamma^{f_2}$, where each $\Gamma^j$ is an AND of distinct set of variables. 
Therefore, we have that $p_{\Gamma^j}^{(b)}$ is the same as in the case of Amano's proof, i.e. $p_{\Gamma^j}^{(b)}  = p_1^{(b)}$. Recall that $p_1^{(b)}$ is equal to $(\frac{1-\delta}{2})^m$ if $b=0$ and it is equal to $(\frac{1+\delta}{2})^m$ if $b=1$. Let $q^{(0)}$ ($q^{(1)}$) denote $\frac{1-\delta}{2}$ (respectively, $\frac{1+\delta}{2}$). 


\begin{align*}
\Delta_{\Gamma}^{(b)} & = \sum_{\substack{j < k:\\ \Vars(\Gamma^j)\cap \Vars(\Gamma^k)\neq \emptyset}} \prob{\bm{x}\sim \mu_b^n}{(\Gamma^j(\bm{x}) =1) \wedge (\Gamma^k(\bm{x}) =1)}. \\
& = \sum_{r=1}^\ell \sum_{\substack{j < k:\\ |\Vars(\Gamma^j)\cap \Vars(\Gamma^k)| = r}} \prob{\bm{x}\sim \mu_b^n}{(\Gamma^j(\bm{x}) =1) \wedge (\Gamma^k(\bm{x}) =1)}
\end{align*}

As $\Gamma^j$ and $\Gamma^k$ are both ANDs of size $m$, $\Gamma^j \wedge \Gamma^k$ is an AND of size $(2m-|\Vars(\Gamma^j)\cap \Vars(\Gamma^k)|)$. Therefore, we get 

\begin{align*}
\Delta_{\Gamma}^{(b)} & = \sum_{r=1}^\ell \sum_{\substack{j < k:\\ |\Vars(\Gamma^j)\cap \Vars(\Gamma^k)| = r}} \prob{\bm{x}\sim \mu_b^n}{((\Gamma^j \wedge \Gamma^k)(\bm{x}) =1)} \\
& = \sum_{r=1}^\ell \sum_{\substack{j < k:\\ |\Vars(\Gamma^j)\cap \Vars(\Gamma^k)| = r}} (q^{(b)})^{2m-r} \\
& = \sum_{r=1}^\ell \sum_{\substack{j < k:\\ |\Vars(\Gamma^j)\cap \Vars(\Gamma^k)| = r}} \frac{(p_1^{(b)})^2}{(q^{(b)})^r} \\
& = \sum_{r=1}^\ell \frac{(p_1^{(b)})^2}{(q^{(b)})^r} \cdot |\{(j,k) \mid j < k \text{ and } |\Vars(\Gamma^j)\cap \Vars(\Gamma^k)| = r\}| \\
\end{align*}

From the construction of the formula, we know that $|\{(j,k) \mid j < k \text{ and } |\Vars(\Gamma^j)\cap \Vars(\Gamma^k)| = r\}| \leq \eta^r f_2^2$. We can also bound $1/q^{(b)}$ by a small constant, say $3$. 

Therefore, we can simplify the above equation as follows:
\begin{align*}
\Delta_{\Gamma}^{(b)} & \leq \sum_{r=1}^\ell {(p_1^{(b)})^2} \cdot 3^r  \cdot \eta^r f_2^2\\
& = (p_1^{(b)})^2 \cdot f_2^2 \sum_{r=1}^\ell 3^r  \cdot \eta^r \\
& \leq (p_1^{(b)})^2 \cdot f_2^2 \cdot 4 \cdot \eta \\
\end{align*}
The last inequality comes from summing up a geometric series. Now using Observation~\ref{obs:pifi} we get that $p_1^{(b)}\cdot f_2 \leq 50m$. Hence, we get 
$\Delta_{\Gamma}^{(b)}  \leq (p_1^{(b)})^2 \cdot f_2^2 \cdot 4 \cdot \eta  \leq (50m)^2 \cdot 4 \eta = (100m)^2 \cdot \eta$. This proves the bound on $\Delta_{\Gamma}^{(b)}$ in the base case. 

We now prove the bounds claimed for $p_\Gamma^{(b)}$ in the base case. When $i=2$, $\beta=0$, hence $p_{\Gamma}^{(b)} = \prob{\bm{x}\sim \mu_b^n}{\Gamma(\bm{x}) = 0}.$ By Janson's inequality (Theorem~\ref{thm:janson}), we get the following bounds on the value of $p_\Gamma^{(b)}$. 

$$\prod_{j=1}^{f_2} (1-p_{\Gamma^j}^{(b)}) \leq p_\Gamma^{(b)} \leq \prod_{j=1}^{f_2} (1-p_{\Gamma^j}^{(b)}) \cdot \exp(2\cdot \Delta_{\Gamma}^{(b)}).$$ 

Recall that $p_{\Gamma^j}^{(b)} = p_1^{(b)}$ as we are in the base case. Also, from Equation~(\ref{eq:pivspi-1}) we have that $(1-p_1^{(b)})^{f_2} = p_2^{(b)}$. Therefore, we get
\begin{align*}
p_2^{(b)}  \leq p_\Gamma^{(b)} & \leq p_2^{(b)} \cdot \exp(2\Delta_\Gamma^{(b)}) & \\
& \leq p_2^{(b)} \cdot (1+4 \cdot \Delta_\Gamma^{(b)})  & \text{Using (\ref{eq:exp-ineq}) (d)}\\
& \leq p_2^{(b)} \cdot (1+4 \cdot(C_4m)^2 \cdot \eta ) & \\
& \leq p_2^{(b)} \cdot (1+ (C_3m)^2 \cdot \eta ) & \\
\end{align*}
This finishes the proof of the base case. 

\paragraph{Inductive case, i.e. $i \geq 3$:} We now proceed to proving the inductive case. Assume that the statement holds for $(i-1)$. Let $\Gamma$ be a subformula at depth $i$ which is OR/AND of subformulas $\Gamma^1, \Gamma^2, \ldots, \Gamma^{f_i}$ each of depth $(i-1)$. From the definition of $\Delta_\Gamma^{(b)}$, we get the following:
\begin{align*}
\Delta_{\Gamma}^{(b)} & = \sum_{\substack{j < k:\\ \Vars(\Gamma^j)\cap \Vars(\Gamma^k)\neq \emptyset}} \prob{\bm{x}\sim \mu_b^n}{(\Gamma^j(\bm{x}) =1-\beta) \wedge (\Gamma^k(\bm{x}) =1-\beta)}. \\
& = \sum_{r=1}^\ell \sum_{\substack{j < k:\\ |\Vars(\Gamma^j)\cap \Vars(\Gamma^k)| = r}} \prob{\bm{x}\sim \mu_b^n}{(\Gamma^j(\bm{x}) =1-\beta) \wedge (\Gamma^k(\bm{x}) =1-\beta)}
\end{align*}

Let $t_r$ denote the maximum value of $\prob{\bm{x}\sim \mu_b^n}{(\Gamma^j(\bm{x}) =1-\beta) \wedge (\Gamma^k(\bm{x}) =1-\beta)}$, where the maximum is taken over $j < k$ such that  $|\Vars(\Gamma^j)\cap \Vars(\Gamma^k)| = r$. Then we get

\begin{align}
\Delta_{\Gamma}^{(b)} & \leq \sum_{r=1}^\ell t_r \cdot  |\{(j,k) \mid j < k \text{ and } |\Vars(\Gamma^j)\cap \Vars(\Gamma^k)| = r\}| \nonumber \\
 & \leq \sum_{r=1}^\ell t_r \cdot \eta^r \cdot f_i^2 \label{eq:tr}
\end{align}

Let us now bound $t_r$, which we will do by using the construction parameters and the inductive hypothesis. Fix any $j < k$. We have

\begin{equation}
\label{eq:GjANDGk}
\prob{\bm{x}\sim \mu_b^n}{(\Gamma^j(\bm{x}) =1-\beta) \wedge (\Gamma^k(\bm{x}) =1-\beta)} = \prob{\bm{x}\sim \mu_b^n}{\Gamma^j(\bm{x}) =1-\beta} \cdot \prob{\bm{x}\sim \mu_b^n}{(\Gamma^k(\bm{x}) =1-\beta) | (\Gamma^j(\bm{x}) =1-\beta)}.
\end{equation}

As $\Gamma^j$ is a formula of depth $i-1$ and $i-1\equiv (1-\beta)\pmod{2}$, using the induction hypothesis, we can upper bound the quantity $\prob{\bm{x}\sim \mu_b^n}{\Gamma^j(\bm{x}) =1-\beta} $. We get
\begin{equation}
\label{eq:Gj}
\prob{\bm{x}\sim \mu_b^n}{\Gamma^j(\bm{x}) =1-\beta} = p_{\Gamma^j}^{(b)} \leq p_{i-1}^{(b)} \cdot (1+\eta \cdot (C_3m)^{i-1}) = p_{i-1}^{(b)} (1+o(1)). 
\end{equation}

We now analyse the second term on the right hand side of Equation~(\ref{eq:GjANDGk}). From the construction of the formula, we know that for any $y \in \Vars(\Gamma^j)$,
the variable $y$ appears in at most $\gamma \cdot f_{i-1}$ many
depth-$(i-2)$ subformulas of $\Gamma^k$. Since
$|\Vars(\Gamma^j)\cap \Vars(\Gamma^k)| = r$, the number of
depth-$(i-2)$ subformulas $T$ of $\Gamma^k$ that contain some variable from $\Gamma^j$
is at most $\gamma \cdot f_{i-1} \cdot r$ which is at most
$\gamma \cdot f_{i-1} \cdot \ell$, as $r \leq \ell$.

Let us construct a formula $\Phi^k$ from $\Gamma^k$ by deleting all the depth-$(i-2)$ subformulas containing some variable from $\Gamma^j$. Then we get 

\begin{align}
\prob{\bm{x}\sim \mu_b^n}{(\Gamma^k(\bm{x}) =1-\beta) | (\Gamma^j(\bm{x}) =1-\beta)} & \leq \prob{\bm{x}\sim \mu_b^n}{(\Phi^k(\bm{x}) =1-\beta) | (\Gamma^j(\bm{x}) =1-\beta)} \nonumber \\
& = \prob{\bm{x}\sim \mu_b^n}{(\Phi^k(\bm{x}) =1-\beta)} \label{eq:phi}
\end{align}
The first inequality follows from the fact that $\Phi^k$ was constructed by removing some subformulas of depth-$(i-2)$ from $\Gamma^k$, and this can only increase the probability of taking value $1-\beta$. The equality follows from the fact that $\Phi^k$ and $\Gamma^j$ share no variables in common and hence the events $(\Phi^k(\bm{x}) =1-\beta)$  and $ (\Gamma^j(\bm{x}) =1-\beta)$ are independent. 

Let $\Gamma^{k,1}, \Gamma^{k,2}, \ldots, \Gamma^{k,f_{i-1}}$ be the depth-$(i-2)$ subformulas of $\Gamma^k$. By ordering the variables if necessary, let $\Gamma^{k,1}, \Gamma^{k,2}, \ldots, \Gamma^{k,f_{i-1}-T}$ be the depth-$(i-2)$ subformulas of $\Phi^k$. 

%

We will show below that
\begin{equation}
\label{eq:phi-final}
\prob{\bm{x}\sim \mu_b^n}{(\Phi^k(\bm{x}) =1-\beta)} \leq \prob{\bm{x}\sim \mu_b^n}{(\Gamma^k(\bm{x}) =1-\beta)} \cdot (1+o(1)).
\end{equation}

Suppose we have this then we will proceed as follows. 

\begin{equation}
\label{eq:phi2}
\prob{\bm{x}\sim \mu_b^n}{(\Phi^k(\bm{x}) =1-\beta)} \leq {\prob{\bm{x}\sim \mu_b^n}{\Gamma^k(\bm{x})=1-\beta}} \cdot (1+o(1))  \leq p_{i-1}^{(b)} \cdot (1+o(1))
\end{equation}
 Here the last inequality is obtained by using the induction hypothesis for $\Gamma^k$. Now using (\ref{eq:Gj}), (\ref{eq:phi}), and (\ref{eq:phi2}) in (\ref{eq:GjANDGk}) we get 

\begin{align*}
\prob{\bm{x}\sim \mu_b^n}{(\Gamma^j(\bm{x}) =1-\beta) \wedge (\Gamma^k(\bm{x}) =1-\beta)}&  \leq (p_{i-1}^{(b)} (1+o(1))) \cdot (p_{i-1}^{(b)} (1+o(1))) \\
& \leq (p_{i-1}^{(b)})^2 (1+o(1))
\end{align*}
Since the above holds for all $j < k$ such that $|\Vars(\Gamma^j)\cap \Vars(\Gamma^k)| = r,$ this gives us a bound on $t_r$. Using this in (\ref{eq:tr}), we get

\begin{align*}
\Delta_{\Gamma}^{(b)}  & \leq  \sum_{r=1}^\ell  (p_{i-1}^{(b)})^2 \cdot (1+o(1)) \cdot \eta^r \cdot f_i^2 \\
& = (p_{i-1}^{(b)})^2 \cdot f_i^2 \cdot (1+o(1)) \cdot \sum_{r=1}^\ell \eta^r \\
& \leq (50m)^2 \cdot 2 \eta
\end{align*}
Here the last inequality is by applying Observation~\ref{obs:pifi} and by summing a geometric series. This therefore proves the inductive bound on $\Delta_{\Gamma}^{(b)}$ assuming (\ref{eq:phi-final}). 

In order to prove (\ref{eq:phi-final}), we note that by using Janson's inequality (Theorem~\ref{thm:janson}) for $\Phi^k,$ we get that

\begin{align*}
\prob{\bm{x}\sim \mu_b^n}{(\Phi^k(\bm{x}) =1-\beta)} \leq \prod_{u \leq f_{i-1}-T} (1-p_{\Gamma^{k,u}}^{(b)}) \cdot \exp(2\Delta_{\Phi^k})
\end{align*}

Also observe (Theorem~\ref{thm:janson}) that $\prob{\bm{x}\sim \mu_b^n}{(\Gamma^k(\bm{x}) =1-\beta)}$ is lower bounded by  $\prod_{u \leq f_{i-1}} (1-p_{\Gamma^{k,u}}^{(b)})$. Therefore, we get $$\prod_{u \leq f_{i-1}-T} (1-p_{\Gamma^{k,u}}^{(b)}) \leq  \frac{\prob{\bm{x}\sim \mu_b^n}{\Gamma^k(\bm{x})=1-\beta}}{\prod_{u > f_{i-1}-T} (1-p_{\Gamma^{k,u}}^{(b)})}.$$ Now, we have $\Delta_{\Phi^k}^{(b)} \leq \Delta_{\Gamma^k}^{(b)}$ by the definitions of these quantities and the fact that $\Phi^k$ is obtained from $\Gamma^k$ by removing some depth-$(i-2)$ subformulas. Also, by the induction hypothesis, we have $\Delta_{\Gamma^k}^{(b)} \leq (C_4m)^2 \eta$. As $\eta = 1/m^{10d}$, we get that $\Delta_{\Phi^k} = o(1)$. Hence, $\exp(2\Delta_{\Phi^k}) = \exp(o(1)) \leq (1+o(1))$. Putting these together, we obtain the following inequality.

\begin{equation}
\label{eq:phi1}
\prob{\bm{x}\sim \mu_b^n}{(\Phi^k(\bm{x}) =1-\beta)} \leq \frac{\prob{\bm{x}\sim \mu_b^n}{\Gamma^k(\bm{x})=1-\beta}}{\prod_{u > f_{i-1}-T} (1-p_{\Gamma^{k,u}}^{(b)})} \cdot (1+o(1))
\end{equation}

Now using the induction hypothesis for $p_{\Gamma^{k,u}}^{(b)}$, we get $p_{\Gamma^{k,u}}^{(b)} \leq (1+o(1))\cdot p_{i-2}^{(b)} \leq 2 \cdot p_{i-1}^{(b)}$. Using this bound on the value of $p_{\Gamma^{k,u}}^{(b)}$, we get the following lower bound on $\prod_{u> f_{i-1}-T} (1-p_{\Gamma^{k,u}}^{(b)})$. 

\begin{align*}
\prod_{u> f_{i-1}-T} (1-p_{\Gamma^{k,u}}^{(b)}) & \geq (1-2p_{i-2}^{(b)})^T \\
& \geq 1-2  \cdot T \cdot p_{i-2}^{(b)} \\
& \geq 1-2 \cdot \gamma \cdot f_{i-1} \cdot \ell \cdot p_{i-2}^{(b)} \\
& \geq (1-o(1))
\end{align*}
The third inequality comes from the upper bound on the value of $T$ argued above. Using Observation~\ref{obs:pifi} we get that $f_{i-1}\cdot p_{i-2}^{(b)} \leq 50m$. From our choice of parameters, $\gamma = 1/m^3$ and $\ell \leq m$. Therefore,  we get $\gamma \cdot \ell \cdot f_{i-1}p_{i-2}^{(b)} \leq (1/m^3) \cdot m \cdot 50m = o(1)$. This gives the last inequality above. Putting it together, this gives is (\ref{eq:phi-final}). This finishes the proof of part 2 in Lemma~\ref{lem:main}.

We now proceed to proving the inductive step for part 1 of Lemma~\ref{lem:main}. The proof is very similar to the proof of the analogous statement in the base case. We give the details for the sake of completeness. Using Janson's inequality, we get 

\begin{equation}
\label{eq:pGamma}
\prod_{j=1}^{f_i} (1-p_{\Gamma^j}^{(b)}) \leq p_\Gamma^{(b)} \leq \prod_{j=1}^{f_i} (1-p_{\Gamma^j}^{(b)}) \cdot \exp(2\cdot \Delta_{\Gamma}^{(b)})
\end{equation}

Using (\ref{eq:exp-ineq}) we get $p_\Gamma^{(b)}  \geq \prod_{j=1}^{f_i} (1-p_{\Gamma^j}^{(b)}) \geq \exp(-\sum_{j\leq f_i}p_{\Gamma^j}^{(b)} - \sum_{j\leq f_i} (p_{\Gamma^j}^{(b)})^2)$. To lower bound this quantity, we will first upper bound $p_{\Gamma^j}^{(b)}$. By using the induction hypothesis, we get $p_{\Gamma^j}^{(b)} \leq p_{i-1}^{(b)}(1+\eta \cdot (C_3m)^{i-1})$. Using this, we get $\sum_{j\leq f_i}p_{\Gamma^j}^{(b)} \leq f_i \cdot p_{i-1}^{(b)}(1+\eta \cdot (C_3m)^{i-1})$. 

We will also show that $ \sum_{j\leq f_i} (p_{\Gamma^j}^{(b)})^2$ is negligible. For that observe the following:

\begin{align*}
 \sum_{j\leq f_i} (p_{\Gamma^j}^{(b)})^2&  \leq f_i \cdot (p_{i-1}^{(b)}(1+\eta \cdot (C_3m)^{i-1}))^2 \\
& \leq 4 \frac{(f_i \cdot p_{i-1}^{(b)})^2}{f_i} \leq \frac{O(m^2)}{\lceil m\cdot 2^m\cdot \ln 2\rceil} \leq \eta \cdot (C_3m)^{i-1}
\end{align*}
Here the second inequality comes from the fact that $(1+\eta \cdot (C_3m)^{i-1})) \leq 2$. The other inequalities easily follow from our choice of parameters and Observation~\ref{obs:pifi}. 

\begin{align}
p_\Gamma^{(b)}  & \geq \exp(-\sum_{j\leq f_i}p_{\Gamma^j}^{(b)} - \sum_{j\leq f_i} (p_{\Gamma^j}^{(b)})^2) \nonumber \\
& \geq \exp(-f_i \cdot p_{i-1}^{(b)}(1+\eta \cdot (C_3m)^{i-1}) - \eta \cdot (C_3m)^{i-1} ) \nonumber \\
& = \exp(-f_i \cdot p_{i-1}^{(b)} - \eta \cdot (C_3m)^{i-1}\cdot (f_ip_{i-1}^{(b)}+1)) \nonumber \\
& \geq (1-p_{i-1}^{(b)})^{f_i} (1- (\eta \cdot (C_3m)^{i-1} (f_ip_{i-1}^{(b)}+1)) \label{eq:exp}\\
& \geq p_i^{(b)} (1-(\eta \cdot (C_3m)^{i-1} (50m+1)) \label{eq:ih}\\
& \geq p_i^{(b)} (1-(\eta \cdot (C_3m)^{i-1} C_3m)) \nonumber\\
& = p_i^{(b)} (1-(\eta \cdot (C_3m)^{i}))\nonumber 
\end{align}
Here, the above inequalities can be obtained primarily by simple rearrangement of terms. The inequality (\ref{eq:exp}) uses (\ref{eq:exp-ineq}), while inequality (\ref{eq:ih}) uses the induction hypothesis and Observation~\ref{obs:pifi}. 
This proves the desired lower bound on $p_\Gamma^{(b)}$. Now we prove the upper bound. 

\begin{align}
p_\Gamma^{(b)}  & \leq \prod_{j \leq f_i} (1-p_{\Gamma^j}^{(b)}) \exp(\Delta_\Gamma^{(b)}) \nonumber \\
& \leq \exp(-\sum_{j \leq f_i} p_{\Gamma^j}^{(b)} + 2 \Delta_\Gamma^{(b)}) \nonumber \\
& \leq \exp(-p_{i-1}^{(b)} (1-\eta\cdot (C_3 m)^{i-1})\cdot f_i + 2 \Delta_{\Gamma}^{(b)}) \nonumber \\
& \leq \exp(-p_{i-1}^{(b)}f_i) \exp\left(p_{i-1}^{(b)}\cdot \eta\cdot (C_3 m)^{i-1}\cdot f_i + 2 \Delta_{\Gamma}^{(b)}\right) \nonumber \\
& = \exp(-p_{i-1}^{(b)})^{f_i} \exp\left(p_{i-1}^{(b)}\cdot \eta\cdot (C_3 m)^{i-1}\cdot f_i + 2 \Delta_{\Gamma}^{(b)}\right) \nonumber \\
& \leq \left((1-p_{i-1}^{(b)}) \cdot \exp((p_{i-1}^{(b)})^2)\right)^{f_i} \cdot  \exp\left(p_{i-1}^{(b)}\cdot \eta\cdot (C_3 m)^{i-1}\cdot f_i + 2 \Delta_{\Gamma}^{(b)}\right)  \label{eq:expb}\\
& = (1-p_{i-1}^{(b)})^{f_i} \cdot \exp\left((p_{i-1}^{(b)})^2\cdot f_i + p_{i-1}^{(b)}\cdot \eta\cdot (C_3 m)^{i-1}\cdot f_i + 2 \Delta_{\Gamma}^{(b)}\right) \nonumber \\
& \leq p_i^{(b)} \cdot \exp\left(\eta\cdot (C_3 m)^{i-1} + 50 m \cdot \eta\cdot (C_3 m)^{i-1} + 2 \eta \cdot (C_3 m)^{i-1}\right) \label{eq:ihu}\\
& \leq p_i^{(b)} \cdot \exp\left( \eta \cdot (C_3m)^{i-1} \cdot (50m+3)\right) \nonumber \\
& \leq p_i^{(b)} \cdot (1+ 2 \cdot \eta \cdot (C_3m)^{i-1} \cdot (50m+3)) \label{eq:1+2x}\\
& \leq p_i^{(b)} \cdot (1+\eta \cdot (C_3m)^i) \nonumber
\end{align}
Most inequalities above are obtained by simple rearrangement of terms. Inequality (\ref{eq:expb}) is obtained by applying the inequality (b) from (\ref{eq:exp-ineq}). Inequality (\ref{eq:ihu}) is obtained by applying (\ref{eq:pivspi-1}), by using Observation~\ref{obs:pifi}, and by using the fact that $f_i \cdot (p_{i-1}^{(b)})^2 \leq \eta \cdot (C_3m)^{i-1}$. Finally, (\ref{eq:1+2x}) is obtained by using inequalities (d) and (e) of (\ref{eq:exp-ineq}). 
This completes the proof of part 1 of Lemma~\ref{lem:main}. 

\end{proof}

\section{Lower bounds for the Coin Problem}
\label{sec:lbds}

In this section, we prove Theorem~\ref{thm:lb-intro}. We start with a special case of the theorem (that we call the monotone case) the proof of which is shorter and which suffices for the application to the Fixed-Depth Size-Hierarchy theorem (Theorem~\ref{thm:size-hie}). We then move on to the general case.

The special case is implicit in the results of O'Donnell and Wimmer~\cite{OW} and Amano~\cite{Amano}, but we prove it below for completeness.

\subsection{The monotone case}
\label{sec:monotone}

In this section, we prove a near-optimal size lower bound (i.e. matching the upper bound construction from Theorem~\ref{thm:main-intro}) on the size of any $\AC^0[\oplus]$ formula computing any \emph{monotone} Boolean function solving the $\delta$-coin problem. Observe that this already implies Theorem~\ref{thm:size-hie}, since the formula $F_n$ from the statement of Theorem~\ref{thm:size-hie} computes a monotone function.

Let $g:\{0,1\}^N\rightarrow \{0,1\}$ be any \emph{monotone} Boolean function solving the $\delta$-coin problem. Note that the monotonicity of $g$ implies that for all $\alpha\in [0,(1-\delta)/2]$ and $\beta \in [(1+\delta)/2,1],$ we have
\begin{equation}
\label{eq:monotonicity}
\prob{\bm{x}\sim D_{\alpha}^N}{g(\bm{x}) =1}\leq 0.1 \text{   and    } \prob{\bm{x}\sim D_{\beta}^N}{g(\bm{x}) =1}\geq 0.9.
\end{equation}

Let $F$ be any $\AC^0[\oplus]$ formula of size $s$ and depth $d$ computing $g$. We will show that $s\geq \exp(d\cdot\Omega(1/\delta)^{1/(d-1)}).$ 

Our main tool is the following implication of the results of Razborov~\cite{Razborov}, Smolensky~\cite{Smolensky93}, and Rossman and Srinivasan~\cite{RossmanS}.
\begin{theorem}
\label{thm:RS}
Let $F'$ be any $\AC^0[\oplus]$ formula of size $s'$ and depth $d$ with $n$ input bits that agrees with the $n$-bit Majority function in at least a $0.75$ fraction of its inputs. Then, $s' \geq \exp(d\cdot\Omega(n)^{1/2(d-1)}).$
\end{theorem}

We will use the above theorem to lower bound $s$ (the size of $F$) by using $F$ to construct a formula $F'$ of size at most $s$ that agrees with the Majority function on $n = \Theta(1/\delta^2)$ bits at a $0.8$ fraction of its inputs. Theorem~\ref{thm:RS} then implies the result. 

We now describe the construction of $F'$. Let $n = \lfloor (1/100\delta^2)\rfloor$. We start by defining a \emph{random} formula $\bm{F}''$ on $n$ inputs as follows. On input $x = (x_1,\ldots, x_n)\in \{0,1\}^n,$ define $\bm{F}''(x)$ to be $F(x_{\bm{i}_1},\ldots,x_{\bm{i}_N})$ where $\bm{i}_1,\ldots,\bm{i}_N$ are chosen i.u.a.r. from $[n].$

We make the following easy observation. For any $x\in \{0,1\}^n$ and for $\alpha = |x|/n,$
\begin{equation}
\label{eq:F''acc}
\prob{\bm{F}''}{\bm{F}''(x) = 1} = \prob{\bm{y}\sim D_{\alpha}^N}{F(\bm{y}) = 1} = \prob{\bm{y}\sim D_{\alpha}^N}{g(\bm{y}) = 1}.
\end{equation}

In particular, from (\ref{eq:monotonicity}), we see that if $\alpha \leq (1-\delta)/2$ or $\alpha \geq (1+\delta)/2,$ we have $\prob{\bm{F}''}{\bm{F}''(x) \neq \Maj_n(x)}\leq 0.1.$ As a result we get
\begin{align*}
\prob{\bm{x}\sim \{0,1\}^n,\bm{F}''}{\bm{F}''(\bm{x}) \neq \Maj_n(\bm{x})} &= \avg{\bm{x}}{\prob{\bm{F}''}{\bm{F}''(\bm{x}) \neq \Maj_n(\bm{x})}}\\
&\leq \prob{\bm{x}}{|\bm{x}|/n \in ((1-\delta)/2,(1+\delta)/2)}\\ &\ \ + \max_{\alpha\not\in [(1-\delta)/2,(1+\delta)/2)]} \prob{\bm{F}''}{\bm{F}''(\bm{x}) \neq \Maj_n(\bm{x})\ |\ |\bm{x}| = \alpha n}\\
&\leq \prob{\bm{x}}{|\bm{x}|/n \in ((1-\delta)/2,(1+\delta)/2)} + 0.1.
\end{align*}

By Stirling's approximation, it follows that for any $i\in [n]$, $\prob{\bm{x}}{|\bm{x}|=i}\leq \binom{n}{\lfloor n/2\rfloor}/2^n \leq 1/\sqrt{n}.$ Hence, by a union bound, we have $\prob{\bm{x}}{|\bm{x}|/n\in ((1-\delta)/2,(1+\delta)/2)}\leq (\delta n)\cdot 1/\sqrt{n}\leq \delta\sqrt{n} \leq 0.1.$ Plugging this in above, we obtain
\[
\prob{\bm{x}\sim \{0,1\}^n, \bm{F}''}{\bm{F}''(\bm{x}) \neq \Maj_n(\bm{x})} \leq 0.2.
\]
By an averaging argument, there is a fixed choice of $\bm{F}''$, which we denote by $F'$, that agrees with the Majority function $\Maj_n$ on a $0.8$ fraction of all inputs. Note that $F' = F(x_{i_1},\ldots,x_{i_N})$ for some choices of $i_1,\ldots,i_N\in [n]$. Hence, $F'$ is a circuit of depth $d$ and size at most $s$. 

Theorem~\ref{thm:RS} now implies the lower bound on $s$.  

\subsection{The general case}
\label{sec:lbd-gen}

In this section, we prove a general lower bound on the size of any $\AC^0[\oplus]$ formula that solves the coin problem (not necessarily by computing a monotone function). The main technical result is the following theorem about polynomials that solve the coin problem.

\begin{theorem}
\label{thm:lbd-polys}
Let $g\in \F_2[x_1,\ldots,x_N]$ solve the $\delta$-coin problem. Then, $\deg(g) = \Omega(1/\delta).$
\end{theorem}

Given the above result, it is easy to prove Theorem~\ref{thm:lb-intro} in its general form.

\begin{proof}[Proof of Theorem~\ref{thm:lb-intro}]
Assume that $F$ is an $\AC^0[\oplus]$ formula $F$ of size $s$ and depth $d$ on $N$ inputs that solves the $\delta$-coin problem.

Building on Razborov~\cite{Razborov}, Rossman and Srinivasan~\cite{RossmanS} show that for any such $\AC^0[\oplus]$ formula $F$ of size $s$ and depth $d$ and any probability distribution $\mu$ on $\{0,1\}^N$, there exists a polynomial $P\in \F[x_1,\ldots,x_N]$ of degree $O((\log s)/d)^{d-1}$ such that
\[
\prob{\bm{x}\sim \mu}{P(\bm{x})\neq F(\bm{x})}\leq 0.05.
\] 
Taking $\mu = (\mu_0^N + \mu_1^N)/2$, we have for each $b\in \{0,1\},$ the above polynomial $P$ satisfies
\begin{equation}
\label{eq:PvsF}
\prob{\bm{x}\sim \mu_b^N}{P(\bm{x})\neq F(\bm{x})}\leq 2\prob{\bm{x}\sim \mu}{P(\bm{x})\neq F(\bm{x})}\leq 0.1.
\end{equation}

In particular, if $F$ solves the $\delta$-coin problem, then $P$ solves the $\delta$-coin problem with error at most $0.2$. By Fact~\ref{fac:err_redn} applied with $t$ being a suitably large constant, it follows that there is a polynomial $Q\in \F[x_1,\ldots,x_N]$ that solves the $\delta$-coin problem (with error at most $0.1$) and satisfies $\deg(Q) \leq t\cdot\deg(P) = O(\deg(P)).$ By Theorem~\ref{thm:lbd-polys}, it follows that $\deg(Q) = \Omega(1/\delta)$ and hence we have $\deg(P) = \Omega(1/\delta)$ as well.

Since $\deg(P) = O((\log s)/d)^{(d-1)}$, we get $s \geq \exp(\Omega(d\cdot (1/\delta)^{1/(d-1)})).$  	
\end{proof}

We now turn to the proof of Theorem~\ref{thm:lbd-polys}.

\subsubsection{Proof of Theorem~\ref{thm:lbd-polys}}
\label{sec:lbd-polys}

We define a \emph{probabilistic function} to be a random function $\bm{g}:\{0,1\}^N\rightarrow \{0,1\}$, chosen according to some distribution. We say that $\deg(\bm{g}) \leq D$ if this distribution is supported over polynomials (from $\F_2[x_1,\ldots,x_N]$ ) of degree at most $D$. A probabilistic function solves the $\delta$-coin problem with error at most $\varepsilon$ if it satisfies (\ref{eq:cpdefn}), where the probability is additionally taken over the randomness used to sample $\bm{g}.$ If no mention is made of the error, we assume that it is $0.1.$ Note that a standard (i.e. non-random) function is also a probabilistic function of the same degree.

We will prove the stronger statement that any probabilistic function $\bm{g}$ solving the $\delta$-coin problem must have degree $\Omega(1/\delta).$ \footnote{While formally stronger, this statement is more or less equivalent, since given such a probabilistic function $\bm{g}$, one can always extract a deterministic function of the same degree that solves the coin problem with error at most $0.21$ by an averaging argument. Then using error reduction (Fact~\ref{fac:err_redn}), we can obtain a deterministic function with a slightly larger degree that solves the coin problem with error $0.1$.  }

Given a probabilistic function $\bm{g}:\{0,1\}^N\rightarrow \{0,1\},$ we define the \emph{profile of $\bm{g}$}, denoted $\pi_{\bm{g}}$, to be a function $\pi_{\bm{g}}:[0,1]\rightarrow [0,1]$ where 
\[
\pi_{\bm{g}}(\alpha) = \prob{\substack{\bm{g},\\ \bm{x}\sim D_\alpha^N}}{\bm{g}(\bm{x}) = 1}.
\]

Note that since $\bm{g}$ solves the $\delta$-coin problem, we have
\begin{equation}
\label{eq:pig}
\pi_{\bm{g}}((1-\delta)/2) \leq 0.1 \text{  and  } \pi_{\bm{g}}((1+\delta)/2) \geq 0.9.
\end{equation}

The proof of the lower bound on $\deg(\bm{g})$ proceeds in two phases. In the first phase, we use $\bm{g}$ to obtain a probabilistic function $\bm{h}$ (of related degree) which satisfies a stronger criterion than (\ref{eq:pig}): namely that the profile of $\bm{h}$ is small in an \emph{interval} close to $(1-\delta')/2$ and large in an interval close to $(1+\delta')/2$ (for some $\delta'\leq \delta$). In the second phase, we use algebraic arguments~\cite{Smolensky} to lower bound $\deg(\bm{h}),$ which leads to a lower bound on $\deg(\bm{g}).$


%


Let $r,t\in\mathbb{N}$ and $\zeta\in (0,1)$ denote absolute constants that we will fix later on in the proof.

We start the first phase of the proof as outlined above. We iteratively define a sequence of probabilistic functions $(\bm{g}_k)_{k\geq 0}$ where $\bm{g}_k:\{0,1\}^{N_k}\rightarrow \{0,1\}$ solves the $\delta_k$-coin problem where $N_k,\delta_k$ are parameters that are defined below.

\begin{itemize}
\item The function $\bm{g}_0$ is simply the function $\bm{g}.$ Hence, $N_0 = N$ and we can take $\delta_0 = \delta.$
\item Having defined $\bm{g}_k,$ we consider which of the following 3 cases occur.
\begin{itemize}
\item Case 1: There is an $\alpha\in ((1-\delta_k)/2, (1-\delta_k)/2 + \delta_k/r]$ such that $\pi_{\bm{g}_k}(\alpha) \geq 0.4.$
\item Case 2: There is an $\beta\in [(1+\delta_k)/2 - \delta_k/r, (1+\delta_k)/2)$ such that $\pi_{\bm{g}_k}(\beta) \leq 0.6.$
\item Case 3: Neither Case 1 nor Case 2 occur. In this case, the sequence of probabilistic functions ends with $\bm{g}_k.$
\end{itemize}
\item If Case 1 or Case 2 occurs, we extend the sequence by defining $\bm{g}_{k+1}$. For simplicity, we assume Case 1 occurs (Case 2 is handled similarly). Note that in this case we have 
\begin{equation}
\label{eq:alpha-pigk}
\pi_{\bm{g}_k}((1-\delta)/2) \leq 0.1 \text{ and } \pi_{\bm{g}_k}(\alpha) \geq 0.4
\end{equation}
for some $\alpha \in ((1-\delta_k)/2, (1-\delta_k)/2 + \delta_k/r]$.

We will need the following technical claim.

\begin{claim}
\label{clm:convex}
Let $\delta',\delta''\in (0,1)$ be such that $\delta''\geq 4\delta'.$ Assume we have $\alpha_1,\alpha_2,\beta_1,\beta_2\in (0,1)$ such that $(1/4)\leq \alpha_1\leq(1/2), \alpha_2 = \alpha_1 +\delta', \beta_1 = (1-\delta'')/2, \beta_2 = (1+\delta'')/2$. Then, there exist $\gamma,\eta\in [0,1]$ such that for each $i\in [2]$, $\alpha_i = \gamma\cdot \beta_i + (1-\gamma)\cdot \eta.$
\end{claim}
\begin{proof}
	We immediately get
	\[
	(\alpha_2-\alpha_1)=\gamma(\beta_2-\beta_1)\quad\implies\quad\gamma=\frac{\alpha_2-\alpha_1}{\beta_2-\beta_1}=\frac{\delta'}{\delta''}\in(0,1/4].
	\]
	Then
	\[
	\eta=\frac{\alpha_1-\beta_1\gamma}{1-\gamma}>\alpha_1-\beta_1\gamma\ge\frac{1-\beta_1}{4}>0.
	\]
	Further
	\[
	\eta=\frac{\alpha_1-\beta_1\gamma}{1-\gamma}\le\frac{4\alpha_1}{3}\le\frac{2}{3}.
	\]
	So $\eta\in(0,2/3]$.
\end{proof}

Applying the above claim to $\alpha_1 = (1-\delta_k)/2, \alpha_2 = \alpha, \delta' = \alpha_2-\alpha_1$ and $\delta'' = 4\delta_k/r,$ we see that there exist $\gamma,\eta\in [0,1]$ such that 
\begin{equation}
\label{eq:gammaetak}
(1-\delta_k)/2 = \gamma\cdot (1-\delta'')/2 + (1-\gamma)\cdot \eta \text{ and }
\alpha = \gamma\cdot (1+\delta'')/2 + (1-\gamma)\cdot \eta
\end{equation}

To define the function $\bm{g}_{k+1},$ we start with an intermediate probabilistic function $\bm{h}_k$ on $N_k$ inputs. On any input $x\in\{0,1\}^{N_k},$ the function $\bm{h}_k$ is defined as follows.

\noindent
$\bm{h}_k(x)$:
\begin{itemize}
\item Sample a random $\bm{b}$ from the distribution $D_{\gamma}^{N_k}$ and $\bm{y}$ from the distribution $D_{\eta}^{N_k}.$
\item Define $\bm{z}\in \{0,1\}^{N_k}$ by $\bm{z}_i = \bm{b}_i\cdot x_i + (1-\bm{b}_i)\cdot \bm{y}_i.$
\item $\bm{h}_k(x)$ is defined to be $\bm{g}_k(\bm{z}).$
\end{itemize}

Note that as each $\bm{z}_i$ is a (random) degree-$1$ polynomial in $x$, the probabilistic function $\bm{h}_k(x)$ satisfies $\deg(\bm{h}_k)\leq \deg(\bm{g}_k).$ 

Also, note that when $\bm{x}$ is sampled from the $D_{(1-\delta'')/2}^{N_k}$ or $D_{(1+\delta'')/2}^{N_k}$, then by (\ref{eq:gammaetak}), $\bm{z}$ has the distribution $D_{(1-\delta_k)/2}^{N_k}$ or $D_\alpha^{N_k}$ respectively. Hence, we get
\begin{align*}
\pi_{\bm{h}_k}((1-\delta'')/2) = \pi_{\bm{g}_k}((1-\delta_k)/2) \leq 0.1 \text{ and }
\pi_{\bm{h}_k}((1+\delta'')/2) = \pi_{\bm{g}_k}(\alpha) \geq 0.4.
\end{align*}

We are now ready to define $\bm{g}_{k+1}.$ Let $\mathrm{Thr}_{t/4}^t:\{0,1\}^t\rightarrow \{0,1\}$ be the Boolean function that accepts inputs of Hamming weight at least $t/4$. We set $N_{k+1} = N_k\cdot t$ and define $\bm{g}_{k+1}:\{0,1\}^{N_{k+1}}\rightarrow \{0,1\}$ by
\[
\bm{g}_{k+1}(x) = \mathrm{Thr}_{t/4}^t(\bm{h}_k^{(1)}(x^{(1)}),\ldots,\bm{h}_k^{(t)}(x^{(t)}))
\]
where $\bm{h}_k^{(1)},\ldots,\bm{h}_k^{(t)}$ are \emph{independent} copies of the probabilistic function $\bm{h}_k$ and $x^{(i)}\in \{0,1\}^{N_k}$ is defined by $x^{(i)}_j = x_{(i-1)N_k + j}$ for each $j\in [N_k]$. Clearly, $\deg(\bm{g}_{k+1})\leq \deg(\mathrm{Thr}_{t/4}^t)\cdot \deg(\bm{h}_k) \leq \deg(\bm{g}_k)\cdot t.$

Note that if $\bm{x}$ is chosen according to the distribution $D_{(1-\delta'')/2}^{N_{k+1}},$ each $\bm{h}_k^{(i)}(\bm{x}^{(i)})$ is a Boolean random variable that is $1$ with probability at most $0.1$. Similarly, if $\bm{x}$ is chosen according to the distribution $D_{(1+\delta'')/2}^{N_{k+1}},$ each $\bm{h}_k^{(i)}(\bm{x}^{(i)})$ is a Boolean random variable that is $1$ with probability at least $0.4$. By a standard Chernoff bound (see e.g. \cite[Theorem 1.1]{DP}), we see that for a large enough absolute constant $t$,
\begin{equation}
\label{eq:gk+1}
\pi_{\bm{g}_{k+1}}((1-\delta'')/2) = \exp(-\Omega(t)) \leq 0.1 \text{  and  }
\pi_{\bm{g}_{k+1}}((1+\delta'')/2) = 1-\exp(-\Omega(t)) \geq 0.9.
\end{equation}
We now fix the value of $t$ so that the above inequalities hold. Note that we have shown the following.

\begin{observation}
\label{obs:gk+1}
$\bm{g}_{k+1}:\{0,1\}^{N_{k+1}}\rightarrow \{0,1\}$ is a probabilistic function that solves the $\delta_k$-coin problem where $N_{k+1} = N_k \cdot t,\deg(\bm{g}_{k+1})\leq \deg(\bm{g}_k)\cdot t$ and $\delta_{k+1} \leq 4\delta_k/r.$
\end{observation}
\end{itemize}

We now argue that, for $r = 10t$, the above process cannot produce an infinite sequence of probabilistic functions. In other words, there is a fixed $k$ such that $\bm{g}_k$ is in neither Case 1 nor Case 2 mentioned above. 

Assume to the contrary that the above process produces an infinite sequence of probabilistic functions. By Observation~\ref{obs:gk+1} and induction, we see that $\bm{g}_k$ is a probabilistic function on at most $N\cdot t^k$ variables solving the $\delta_k$-coin problem for $\delta_k\leq \delta\cdot (4/r)^k.$ We now appeal to the following standard fact.

\begin{fact}[Folklore]
\label{fac:stat-dist}
Let $\delta'\in (0,1)$ and $N'\in \mathbb{N}$ be arbitrary. Then, the statistical distance between $D_{(1-\delta')/2}^{N'}$ and $D_{(1+\delta')/2}^{N'}$ is at most $O(\sqrt{N'}\cdot \delta').$
\end{fact}

Thus, for $\bm{g}_k$ to able to solve the $\delta_k$-coin problem with $N_k$ samples, we must have $\sqrt{N_k}\cdot \delta_k\geq \alpha_0$ for some absolute positive constant $\alpha_0.$ On the other hand, this cannot be true for large enough $k$, since $\sqrt{N_k}\delta_k \leq N_k\delta_k \leq N\delta \cdot (4t/r)^k$ and $r\geq 10t.$ This yields a contradiction.

Thus, we have shown that for large enough $k$, the function $\bm{g}_k$ is in neither Case 1 nor Case 2. Equivalently, for any $\alpha\in [1/2-\delta_k,1/2-\delta_k+\delta_k/r]$ and any $\beta \in [1/2+\delta_k-\delta_k/r, 1/2 + \delta_k],$ we have
\[
\pi_{\bm{g}_k}(\alpha) \leq 0.4 \text{  and  } \pi_{\bm{g}_k}(\beta) \geq 0.6.
\]

Using error reduction as above, we can obtain a probabilistic function that satisfies the above inequalities with parameters $\zeta := \exp(-10r^2)$ and $1-\zeta$ respectively. Set $\ell = 10\lceil \log(1/\zeta)\rceil$ and define $\bm{h}:\{0,1\}^{N_k\cdot \ell}\rightarrow \{0,1\}$ by
\srikanth{I think this value of $\zeta$ should be small enough. Redefine if necessary.}
\[
\bm{h}(x) = \Maj_\ell(\bm{g}_k^{(1)}(x^{(1)}),\ldots,\bm{g}_k^{(t)}(x^{(t)}))
\]
where $\bm{g}_k^{(1)},\ldots,\bm{g}_k^{(t)}$ are \emph{independent} copies of the probabilistic function $\bm{g}_k$ and $x^{(i)}\in \{0,1\}^{N_k}$ is defined by $x^{(i)}_j = x_{(i-1)N_k + j}$ for each $j\in [N_k]$. 

Clearly, $\deg(\bm{h}) \leq \ell\cdot \deg(\bm{g}_k) = O(\deg(\bm{g}_k))$ as $\ell$ is an absolute constant. Further, by the Chernoff bound, $\bm{h}$ satisfies

\begin{equation}
\label{eq:h}
\pi_{\bm{h}}(\alpha) \leq \zeta \text{  and  } \pi_{\bm{h}}(\beta) \geq 1-\zeta.
\end{equation}

This concludes the first phase of the proof. In the second phase, we will show the following lower bound on $\deg(\bm{h})$.

\begin{claim}
\label{clm:h}
$\deg(\bm{h})\geq \Omega(1/\delta_k).$
\end{claim}

Note that this immediately implies the result since we have 
\[
\deg(\bm{g}) = \deg(\bm{g}_0) \geq \frac{\deg(\bm{g}_k)}{t^k} = \Omega\left(\frac{\deg(\bm{h})}{t^k}\right) =  \Omega\left(\frac{1}{\delta_kt^k}\right) \geq \Omega\left(\frac{1}{\delta\cdot (4t/r)^k}\right) \geq \Omega(1/\delta)
\]
where we have used Observation~\ref{obs:gk+1} and the fact that $r \geq 10t.$

It therefore suffices to prove Claim~\ref{clm:h}. We prove this in two steps. 

We start with an extension of a lower bound of Smolensky~\cite{Smolensky} (see also the earlier results of Razborov~\cite{Razborov} and Szegedy~\cite{szegedy}) on the degrees of polynomials approximating the Majority function.\footnote{In~\cite{Smolensky}, Smolensky proves a lower bound for $\mathrm{MOD}_p$ functions. However, the same idea also can be used to prove lower bounds for the Majority function, as observed by Szegedy~\cite{szegedy}.}  Our method is a slightly different phrasing of this bound following the results of Aspnes, Beigel, Furst and Rudich~\cite{ABFR}, Green~\cite{Green} and Kopparty and Srinivasan~\cite{KS}.\footnote{It can be viewed as a `dual' version of Smolensky's proof. Smolensky's standard proof can also be modified to yield this.}

\begin{lemma}[A slight extension of Smolensky's bound]
\label{lem:smol-ext}
Let $h:\{0,1\}^n\rightarrow \{0,1\}$ be a (deterministic) function satisfying the following. There exist integers $D < R < n/2$ such that $E_h^R$ defined by 
\begin{equation}
\label{eq:EhRdef}
E_h^R = \{ x\in \{0,1\}^n\ |\ h(x) \neq \mathrm{Maj}_{n}(x), |x|\not\in (R,n - R)\}
\end{equation}
satisfies $|E_h^R| <  \binom{n}{\leq (R - D)}.$\footnote{Recall that $\binom{n}{\leq i}$ denotes $\sum_{j=0}^i \binom{n}{j}.$} Then, $\deg(h) > D.$
\end{lemma}

\begin{proof}
Consider the vector space $V_{R-D}$ of all multilinear polynomials of degree $\leq (R - D)$. A generic polynomial $g\in V_{R-D}$ is given by
\[
g(x_1,\ldots,x_n) = \sum_{|S|\leq R-D}a_S \cdot \prod_{i\in S} x_i
\] 
where $a_S\in \F_2$ for each $S$. We claim that there is a \emph{non-zero} $g$ as above that vanishes at \emph{all points} in $E_h^R.$ To see this, note that finding such a $g$ is equivalent to finding a non-zero assignment to the $a_S$ so that the resulting $g$ vanishes at $E_h^R$. Vanishing at any point of $\{0,1\}^n$ gives a linear constraint on the coefficients $a_S$. Since we have $|E_h^R| <  \binom{n}{\leq (R - D)}$, we have a homogeneous linear system with more variables than constraints and hence, there exists a non-zero multilinear polynomial $g$ of degree $\leq (R-D)$ which vanishes on $E_h^R$. Thus, there is a non-zero $g$ as claimed.

Let $B_1 = \{ x\in \{0,1\}^n\ |\ |x| \leq R\}$ and $B_2 = \{ x\in \{0,1\}^n\ |\ |x| \geq n-R\}$. Note that $B_1$ and $B_2$ are both Hamming balls of radius $R$. Let $f$ be the pointwise product of the functions $g$ and $h$. Note that $f$ can be represented as a \emph{multilinear} polynomial of degree at most $\deg(g) + \deg(h)$ (by replacing $x_i^2$ by $x_i$ as necessary in the polynomial $g\cdot h$).

We will need the following standard fact (see e.g.~\cite{KS} for a proof).

\begin{fact}
\label{fac:hamming balls}
Let $P$ be a non-zero multilinear polynomial of degree $\leq d$ in $n$ variables. Then $P$ cannot vanish on a Hamming ball of radius $d$.
\end{fact}

On $B_1$, $g$ vanishes wherever $h$ does not (by definition of $E_h^R$) and therefore $f$ vanishes everywhere. 

On $B_2$, $h$ is non-vanishing wherever $g$ is non-vanishing. But since $B_2$ is a Hamming ball of radius $R > R-D \geq \text{deg}(g)$, Fact~\ref{fac:hamming balls} implies that $g(x_0) \neq 0$ for some point $x_0$ in $B_2$. Therefore $h(x_0) \neq 0$ and so $f(x_0) \neq 0$. In particular, $f$ is a non-zero multilinear polynomial of degree at most $\deg(g) + \deg(h).$

Since $f$ is non-zero and vanishes on $B_1$ which is a Hamming ball of radius $R$, Fact~\ref{fac:hamming balls} implies that $\text{deg}(f) > R$. Therefore $(R-D)+\text{deg}(h) \geq \text{deg}(g)+\text{deg}(h) \geq \text{deg}(f) > R \Rightarrow \text{deg}(h) > D$.
\end{proof}

We now prove Claim~\ref{clm:h}.

\begin{proof}[Proof of Claim~\ref{clm:h}]
The idea is to use $\bm{h}$ to produce a deterministic function $h$ of the same degree to which Lemma~\ref{lem:smol-ext} is applicable. Let $M$ denote $N_k\cdot \ell$, the sample complexity (i.e. number of inputs) of $\bm{h}.$ 

Let $n = \lceil r^2/\delta_k^2\rceil $. Define a probabilistic function $\tilde{\bm{h}}:\{0,1\}^n\rightarrow \{0,1\}$ as follows. On any input $x\in \{0,1\}^n$, we choose $\bm{i}_1,\ldots,\bm{i}_M\in [n]$ i.u.a.r. and set
\[
\tilde{\bm{h}}(x) = \bm{h}(x_{\bm{i}_1},\ldots,x_{\bm{i}_M}).
\]

Clearly, we have $\deg(\tilde{\bm{h}})\leq \deg(\bm{h}).$ Also note that for any $x\in \{0,1\}^n$, we have
\[
\prob{\tilde{\bm{h}}}{\tilde{\bm{h}}(x) = 1} = \pi_{\bm{h}}(|x|/n)
\]
since in this case $(x_{\bm{i}_1},\ldots,x_{\bm{i}_M})$ is drawn from the distribution $D_{|x|/n}^M.$  By (\ref{eq:h}), we have for any $x$ such that $|x|/n\in [1/2-\delta_k,1/2-\delta_k+\delta_k/r]\cup [1/2+\delta_k-\delta_k/r,1/2+\delta_k],$ 
\[
\prob{\tilde{\bm{h}}}{\tilde{\bm{h}}(x) \neq \Maj_n(x)} \leq \zeta.
\]
In particular, for $\bm{x}$ chosen uniformly at random from $\{0,1\}^n$, we have
\[
\prob{\bm{x},\tilde{\bm{h}}}{\tilde{\bm{h}}(\bm{x}) \neq \Maj_n(\bm{x})\ |\ |\bm{x}|/n\in [1/2-\delta_k+\delta_k/r,1/2-\delta_k]\cup [1/2+\delta_k-\delta_k/r,1/2+\delta_k] } \leq \zeta.
\]

We will apply Lemma~\ref{lem:smol-ext} below with $n$ unchanged, $R = \lfloor (1/2-\delta_k + \delta_k/r)\cdot n\rfloor,$ and $D = \lfloor \delta_kn/(2r)\rfloor.$ For these parameters, we have 
\begin{align*}
\avg{\tilde{\bm{h}}}{\frac{|E^R_{\tilde{\bm{h}}}|}{2\binom{n}{\leq R}}} &= \avg{\tilde{\bm{h}}}{\prob{\bm{x}}{\tilde{\bm{h}}(\bm{x}) \neq \Maj_n(\bm{x})\ |\ |\bm{x}|\not\in (R,n-R)}}\\
&=\prob{\bm{x},\tilde{\bm{h}}}{\tilde{\bm{h}}(\bm{x}) \neq \Maj_n(\bm{x})\ |\ |\bm{x}|\not\in (R,n-R)}\\
&\leq \prob{\bm{x},\tilde{\bm{h}}}{\tilde{\bm{h}}(\bm{x}) \neq \Maj_n(\bm{x})\ |\ |\bm{x}|/n\in [1/2-\delta_k,1/2-\delta_k+\delta_k/r]\cup [1/2+\delta_k-\delta_k/r,1/2+\delta_k] }\\
&+ \prob{\bm{x},\tilde{\bm{h}}}{ |\bm{x}|/n\not\in [1/2-\delta_k,1/2+\delta_k]\ |\ |\bm{x}|\not\in (R,n-R)}\\
&\leq \zeta + \prob{\bm{x}}{ |\bm{x}|/n\not\in [1/2-\delta_k,1/2+\delta_k]\ |\ |\bm{x}|\not\in (R,n-R)} \leq \zeta + \frac{2\binom{n}{\leq R'}}{2\binom{n}{\leq R}}
\end{align*}
where $R' = \lfloor (1/2-\delta_k)n\rfloor.$ Hence, we have
\[
\avg{\tilde{\bm{h}}}{|E^R_{\tilde{\bm{h}}}|} \leq 2\zeta\binom{n}{\leq R} + 2\binom{n}{\leq R'}.
\]

By averaging, we can fix a deterministic $h:\{0,1\}^n\rightarrow \{0,1\}$ so that $\deg(h)\leq \deg(\tilde{\bm{h}})$ and we have the same bound for $|E^R_h|.$ By a computation, we show below that the above bound is strictly smaller than $\binom{n}{\leq R - D}.$ Lemma~\ref{lem:smol-ext} then  implies that $\deg(h)\geq D = \Omega(1/\delta_k),$ completing the proof of Claim~\ref{clm:h}. 

It remains to prove only that 
\begin{equation}
\label{eq:cltineq}
2\zeta\binom{n}{\leq R} + 2\binom{n}{\leq R'} \leq \binom{n}{\leq R-D}.
\end{equation}
\srikanth{UV: fill in proof.}

	Recall that $\zeta=e^{-10r^2}$, where $r=10t$ and $t\geq 1$.   
	
	Now let $V_n(\alpha)=\frac{1}{2^n}{n\choose\le n/2-\alpha\sqrt{n}}$ and $\Phi(x)=\int_x^\infty\frac{e^{-t^2/2}}{\sqrt{2\pi}}\,dt$.  Then by the Central Limit Theorem (see, e.g.,~\cite[Chapter VIII, vol II]{Feller}), for any fixed $\alpha$, we have
	\begin{equation}
	\label{eq:clt}
	\lim_{n\to\infty}(V_n(\alpha)-\Phi(2\alpha))=0.
	\end{equation}
	Also, we have the estimates~\cite[Lemma 2, Chapter VII, vol I]{Feller}
	\begin{equation}
	\label{eq:gaussian}
	\frac{1}{\sqrt{2\pi}}\cdot\left(\frac{1}{x}-\frac{1}{x^3}\right)e^{-x^2/2}\le\Phi(x)\le \frac{1}{\sqrt{2\pi}}\cdot \frac{1}{x}e^{-x^2/2},\quad x>0.
	\end{equation}
Note that we have $\delta_k n \geq \delta_k\cdot (r/\delta_k) \cdot \sqrt{n} = r\sqrt{n}.$ So we get
	\begin{align}
	\frac{1}{2^n}\cdot 2\zeta{n\choose\le R}+\frac{1}{2^n}{n\choose\le R'}&\le \frac{1}{2^n}\cdot 2\zeta{n\choose \le n/2-(1-1/r)\delta_kn}+\frac{1}{2^n}\cdot 2{n\choose\le n/2-\delta_kn}\notag\\
	&\le \frac{1}{2^n}\cdot 2\zeta{n\choose \le n/2-(1-1/r)r\sqrt{n}}+\frac{1}{2^n}\cdot 2{n\choose\le n/2-r\sqrt{n}}\notag\\
	&= 2\zeta\cdot V_n(r-1) + V_n(r)\notag\\
	&\leq 2\zeta + V_n(r)\notag\\
	&\leq 2e^{-10r^2} + \frac{1}{\sqrt{2\pi}r}\cdot e^{-r^2/2} + o(1)\label{eq:clt1}
	\end{align}
	where the second-last inequality uses the fact $V_n(\alpha)\leq 1$ for any $\alpha,$ and the final inequality follows from the definition of $\zeta$ and (\ref{eq:clt}) and (\ref{eq:gaussian}) above. Note that the $o(1)$ above goes to $0$ as $n\rightarrow \infty$ (which happens as $\delta \rightarrow 0$).

	Further, we have $\delta_k n \leq \delta_k \cdot ((r/\delta_k) + 1)\cdot \sqrt{n} = (r+o(1))\sqrt{n},$ which yields
	\begin{align}
	\frac{1}{2^n}{n\choose\le R-D}&=\frac{1}{2^n}{n\choose\le \lfloor n/2-(1-1/r)\delta_kn\rfloor-\lfloor(1/2r)\delta_kn\rfloor}\notag\\
	&\ge\frac{1}{2^n}{n\choose \le n/2-(1-1/2r)\delta_kn-1}\notag\\
	&= \frac{1}{2^n}{n\choose \le n/2-(1-1/2r+o(1))\cdot \delta_kn}\notag\\
	&\geq\frac{1}{2^n}{n\choose \le n/2-(1-1/2r+o(1))\cdot (r+o(1))\sqrt{n}} \notag\\
	&= \frac{1}{2^n}{n\choose \le (n/2)- (r-(1/2) + o(1))\cdot\sqrt{n}} \geq V_n(r-(1/4))\qquad \text{(for large enough $n$)}\notag\\
	&\geq \frac{1}{\sqrt{2\pi}} e^{-(r-(1/4))^2/2} \cdot \left(\frac{1}{r-(1/4)} - \frac{1}{(r-1/4)^3}\right) - o(1)\notag\\	
	&\geq  \frac{1}{\sqrt{2\pi}} e^{-(r^2/2) + 2}\cdot \frac{1}{2r} - o(1)
	\label{eq:clt2}
	\end{align}
	where the second-last inequality uses (\ref{eq:clt}) and (\ref{eq:gaussian}) and the final inequality uses the fact that $r\geq 10t \geq 10.$
	From (\ref{eq:clt1}) and (\ref{eq:clt2}), the inequality follows for all large $n$.
\end{proof}

\section{Open Problems}
\label{sec:discussion}

We close with some open problems. 

\begin{itemize}
%

\item We get almost optimal upper and lower bounds on the complexity of the $\delta$-coin problem. In particular, as mentioned in Theorem~\ref{thm:main-intro} and Theorem~\ref{thm:lb-intro}, the upper and lower bounds on the size of $\AC^0[\oplus]$ formulas computing the $\delta$-coin problem are $\exp(O(d(1/\delta)^{1/(d-1)}))$ and $\exp(\Omega(d(1/\delta)^{1/(d-1)}))$, respectively. It may be possible to get even tighter bounds by exactly matching the constants in the exponents in these bounds. The strongest result in this direction would be to give explicit separations between $\AC^0[\oplus]$ formulas of size $s$ and size $s^{1+\varepsilon}$ for any fixed $\varepsilon > 0$ (for $s = n^{O(1)}$, for example).

For circuits, we have a \emph{non-explicit} upper bound of $\exp(O((1/\delta)^{1/(d-1)}))$ due to Rossman and Srinivasan~\cite{RossmanS}. It would be interesting to achieve this upper bound with an explicit family of circuits. 

\item In Theorem~\ref{thm:main-intro} we get a $(1/\delta)^{2^{O(d)}}$ upper bound on the sample complexity of the $\delta$-coin problem. We believe that this can be improved to $O((1/\delta)^2)$. (We get this for depth-$2$ formulas, but not for larger depths.)

\item Finally, can we match the $\AC^0$ size-hierarchy theorem of Rossman~\cite{Rossman-clique} by separating $\AC^0[\oplus]$ circuits of size $s$ and some fixed depth (say $2$) and $\AC^0[\oplus]$ circuits of size $s^{\varepsilon}$ (for some absolute constant $\varepsilon > 0$) and \emph{any} constant depth?

\end{itemize}

\paragraph{Acknowledgements.} The authors thank Paul Beame, Prahladh Harsha, Ryan O'Donnell, Ben Rossman, Rahul Santhanam, Ramprasad Saptharishi, Madhu Sudan, Avishay Tal, Emanuele Viola, and Avi Wigderson for helpful discussions and comments. The authors also thank the anonymous referees of STOC 2018 for their comments.

\bibliographystyle{alphaurl}
\bibliography{coinbib}
\appendix
\section{Omitted Proofs from Section \ref{sec:amano}}
\begin{reptheorem}{thm:amano}
Assume $d\geq 3$ and $F_d$ is defined as in section \ref{sec:amano}. Then, for small enough $\delta$, we have the following.
\begin{enumerate}
\item For $b,\beta\in \{0,1\}$ and each $i\in [d-1]$ such that $i\equiv \beta \pmod{2}$, we have 
\[
p_i^{(b)} = \prob{\bm{x}\sim \mu_b^{N_i}}{F_i(\bm{x}) = \beta}.
\]
In particular, for any $i\in \{2,\ldots,d-2\}$ and any $b\in \{0,1\}$
\begin{equation}
p_i^{(b)} = (1-p_{i-1}^{(b)})^{f_i}.
\end{equation}
\item For $\beta \in \{0,1\}$ and $i\in [d-2]$ such that $i\equiv \beta\pmod{2},$ we have
\begin{align*}
\frac{1}{2^m}(1+\delta_i\exp(-3\delta_i))&\leq  p_i^{(\beta)} \leq \frac{1}{2^m}(1+\delta_i\exp(3\delta_i))\\
\frac{1}{2^m}(1-\delta_i\exp(3\delta_i))&\leq  p_i^{(1-\beta)} \leq \frac{1}{2^m}(1-\delta_i\exp(-3\delta_i))
\end{align*}
\item Say $d-1\equiv \beta \pmod{2}.$ Then
\[
p_{d-1}^{(\beta)} \geq \exp(-C_1m + C_2) \text{ and } p_{d-1}^{(1-\beta)} \leq \exp(-C_1m - C_2)
\]
where $C_2 = C_1/10.$
\item For each $b\in \{0,1\}$, $\prob{\bm{x}\sim \mu^N_b}{F_d(\bm{x}) = 1-b}\leq 0.05.$ In particular, $F_d$ solves the $\delta$-coin problem.
\end{enumerate}
\end{reptheorem}
\begin{proof}  
  \textbf{Proof of (1) (for $i \in [d-2]$) and (2):} We show these by induction on $i$. We start with the base case $i=1$. Each formula at level $1$ computes an AND of $N_1=f_1=m$ many variables. Hence we have:
  \begin{flalign*}
    &&\prob{D_0^{N_1}}{F_1(\bm{x})=1}&=\left(\frac{1-\delta}{2}\right)^{f_1}= \left(\frac{1-\delta}{2}\right)^{m}\le \frac{1}{2^m}\le 0.5&\\
    &&\prob{D_1^{N_1}}{F_1(\bm{x})=1}&=\left(\frac{1+\delta}{2}\right)^{f_1}= \left(\frac{1+\delta}{2}\right)^{m}\le 0.5 &\text{(for small enough $\delta$)}\\
  \end{flalign*}
  This implies $p_1^{(b)}=\prob{D_{(b)}^{N_1}}{F(\bm{x})=1}$. For part (2):  
  \begin{flalign*}
    &&p_1^{(1)}&=\prob{D_1^m}{F_1(\bm{x})=1}=\left(\frac{1+\delta}{2}\right)^m&\\
    &&    &= \frac{1}{2^m}(1+\delta)^m\le \frac{1}{2^m} \exp(\delta m)& \text{By Fact~\ref{fac:exp} (c) and $\delta m = o(1)$}\\
    &&    &\le \frac{1}{2^m}(1+\delta m+\delta^2 m^2) & \text{By Fact~\ref{fac:exp} (d)}\\
    &&    &= \frac{1}{2^m}(1+\delta m(1+\delta m))& \\
    &&    &\leq \frac{1}{2^m}(1+\delta m \exp(\delta m))& \text{By Fact~\ref{fac:exp} (c)}\\
    &&    &= \frac{1}{2^m}(1+\delta_1 \exp(\delta_1))&\\
    &&    &\le \frac{1}{2^m}(1+\delta_1 \exp(3\delta_1))&\\  
    &&  p_1^{(1)}&= \frac{1}{2^m}(1+\delta)^m\ge \frac{1}{2^m}\exp(m(\delta-\delta^2)) & \text{By Fact~\ref{fac:exp} (b) with $x = \delta$}\\
    &&      &\ge \frac{1}{2^m}(1+(m(\delta-\delta^2)))& \text{By Fact~\ref{fac:exp} (c)}\\
    &&      &= \frac{1}{2^m}(1+(m\delta(1-\delta)))\ge \frac{1}{2^m}(1+(m\delta\exp(-2\delta))) & \text{By Fact~\ref{fac:exp} (a),(b)}\\
    &&      &\ge \frac{1}{2^m}(1+(\delta_1\exp(-3\delta_1)))&\\
  \end{flalign*}
 Similarly, we bound $p_1^{(0)}$:
  \begin{flalign*}
    && p_1^{(0)} &= \frac{1}{2^m}(1-\delta)^m \le \frac{1}{2^m}\exp(-m\delta)& \text{By Fact~\ref{fac:exp} (c)}\\
    && &\le \frac{1}{2^m}(1-m\delta(1-m\delta))& \text{By Fact~\ref{fac:exp} (d)}\\
    && &\le \frac{1}{2^m}(1-m\delta\exp(-3m\delta))& \text{By Fact~\ref{fac:exp} (a),(b)} \\
    && &= \frac{1}{2^m}(1-\delta_1\exp(-3\delta_1))&\\
    && p_1^{(0)} &= \frac{1}{2^m}(1-\delta)^m &\\
    && &\ge \frac{1}{2^m}\exp(m(-\delta-\delta^2))& \text{By Fact~\ref{fac:exp} (b)}\\
    && &\ge \frac{1}{2^m}(1-m\delta(1+\delta))&\text{By Fact~\ref{fac:exp} (c)}\\
    && &\ge \frac{1}{2^m}(1-m\delta(1+3m\delta))&\\
    && &\ge \frac{1}{2^m}(1-\delta_1\exp(3\delta_1))&\text{By Fact~\ref{fac:exp} (c)}
  \end{flalign*}
  We now show the inductive step of parts (1) and (2) for $p_{i+1}^{(b)}$. Since the circuit consists of alternating layers of OR gates and AND gates, we obtain $\forall i\in [d-2]$:
  \begin{align*}
    \prob{D_{b}^{N_i}}{F_i(\bm{x})=\beta}=\left(1-\prob{D_{b}^{N_{i-1}}}{F_{i-1}(\bm{x})=(1-\beta)}\right)^{f_i}
  \end{align*}
  Without loss of generality, assume $i\equiv 0\pmod 2$. Then we have:
  \begin{flalign*}
    &&\prob{\bm{x}\sim\mu_0^{N_{i+1}}}{F_{i+1}(\bm{x})=1}&=    \left(1-\prob{\bm{x}\sim\mu_0^{N_{i}}}{F_{i}(\bm{x})=0}\right)^{f_{i+1}}&\\
    && &=\left(1-p_i^{(0)} \right)^{f_{i+1}}&\text{From induction part (1)}\\
    && &\le \exp(-p_i^{(0)} \cdot f_{i+1})&\text{By Fact~\ref{fac:exp} (c)}\\
    && &= \exp(-p_i^{(0)}\ceil{2^mm\ln 2})&\\
    && &\le \exp(-p_i^{(0)}(2^mm\ln 2))&\\
    && &= \exp\left(-\left(\frac{1}{2^m}(1+\delta_i\exp(-3\delta_i))\right)2^mm\ln 2\right)&\text{From induction part(2)}\\
    && &\le \exp\left(-(m\ln 2)(1+\delta_i\exp(-3\delta_i))\right)&\\
    && &\le \exp(-m\ln 2) =\frac{1}{2^{m}}\le 0.5&\\
    &&\text{Similarly, }\prob{\bm{x}\sim\mu_1^{N_{i+1}}}{F_{i+1}(\bm{x})=1}&=(1-p_i^{(1)})^{f_{i+1}}&\\
    && &\le \exp(-p_i^{(1)}\cdot f_{i+1}) = \exp(-p_i^{(1)}\ceil{2^mm\ln 2})&\text{By Fact~\ref{fac:exp} (c)}\\
    && &\le \exp\left(-\left(\frac{1}{2^m}(1-\delta_i\exp(3\delta_i))\right)(2^mm\ln 2)\right)&\\
    && &\le \exp((-m\ln 2)(1-\delta_i(1+3\delta_i+9\delta_i^2)))&\text{By Fact~\ref{fac:exp} (d)}\\
    && &= 2^{-m(1-\delta_i(1+3\delta_i+9\delta_i^2))}\le 0.5&\text{Since $\delta_i = o(1)$}
  \end{flalign*}
  This completes the induction step for part (1). Now we show the bounds for part (2): 
  \begin{flalign*}
    && p_{i+1}^{(1)} &=(1-p_i^{(1)})^{\ceil{2^mm\ln 2}}\le (1-p_i^{(1)})^{m2^m\ln 2}& \\
    && &\le \left( 1-\frac{1}{2^m}(1-\delta_i\exp(3\delta_i))\right)^{m2^m\ln 2}& \text{From induction hypothesis}\\
    && &\le \exp\left(-\frac{1}{2^m}(1-\delta_i\exp(3\delta_i))m2^m\ln 2\right) & \text{By Fact~\ref{fac:exp} (c)}\\
    && &= \exp\left((-m\ln 2) (1-\delta_i\exp(3\delta_i))\right)&\\
    && &= \exp\left(-m\ln 2+\delta_i\exp(3\delta_i)m\ln 2\right)&\\
    && &=\frac{1}{2^m}\exp(\delta_i\exp(3\delta_i)m\ln 2)&\\
    && &=\frac{1}{2^m}\exp(\delta_{i+1}\exp(3\delta_i))&\\
    && &\le \frac{1}{2^m}\left(1+\delta_{i+1}\exp(3\delta_i)(1+\delta_{i+1}\exp(3\delta_i))\right)& \text{By Fact~\ref{fac:exp} (d)}\\
    && &\le \frac{1}{2^m}(1+\delta_{i+1}\exp(3\delta_i)(1+2\delta_{i+1}))& \text{Since $\delta_{i}= o(1)$}\\
    && &\le \frac{1}{2^m}(1+\delta_{i+1}\exp(3\delta_{i}+\delta_{i+1}))& \text{By Fact~\ref{fac:exp} (c)}\\
    && &\le \frac{1}{2^m}(1+\delta_{i+1}\exp(3\delta_{i+1}))&\\
    && p_{i+1}^{(1)}&=(1-p_i^{(1)})^{\ceil{2^mm\ln 2}}&\\
    && &\ge (1-p_i^{(1)})^{2^mm\ln 2+1} &\\
    && &\ge \left(1-\frac{1}{2^m}(1-\delta_i\exp(-3\delta_i))\right)^{m2^m\ln 2+1}&\text{From induction}\\
    && &\ge \exp\left(\left(\frac{-1}{2^m}(1-\delta_i\exp(-3\delta_i))-\frac{1}{2^{2m}}(1-\delta_i\exp(-3\delta_i))^2\right)\left(m2^m\ln 2+1)\right)\right)&\text{By Fact~\ref{fac:exp} (b)}\\
   && &\ge \exp\left(\left(\frac{-1}{2^m}(1-\delta_i\exp(-3\delta_i))-\frac{1}{2^{2m}}\right)(m2^m\ln 2+1)\right)&\\
   && &\ge \frac{1}{2^m}\exp\left( \delta_i\exp(-3\delta_i)m\ln 2 -\frac{m}{2^m}\right)&\\
   && &= \frac{1}{2^m}\exp\left(\delta_{i+1}\exp(-3\delta_{i})-\frac{m}{2^m}\right)&\\
   && &\ge \frac{1}{2^m}\exp\left(\delta_{i+1}\exp(-3\delta_{i})-\delta_{i+1}^2\right)&\text{Using $\delta_{i+1}\geq \frac{1}{m}$}\\
   && &\ge \frac{1}{2^m}\exp\left(\delta_{i+1}(1-3\delta_{i})-\delta_{i+1}^2\right)&\text{By Fact~\ref{fac:exp} (c)}\\
   && &\ge \frac{1}{2^m}\exp\left(\delta_{i+1}(1-2\delta_{i+1})\right)&\\
   && &\ge \frac{1}{2^m}\exp\left(\delta_{i+1}\exp(-3\delta_{i+1})\right)&\text{By Fact~\ref{fac:exp} (a)(b)}\\
   && &\ge \frac{1}{2^m}\left(1+\delta_{i+1}\exp(-3\delta_{i+1})\right)&\text{By Fact~\ref{fac:exp} (c)}\\
 \end{flalign*}
The case of $p_{i+1}^{(0)}$ is similar and hence omitted.\\
 \textbf{Proof of (3):} Assume, without loss of generality, $d-1\equiv 1\pmod 2$. Let $i=d-1$. Then we have:
 \begin{flalign*}
   \delta_{i-1}&=\delta_{d-2}=\delta m(m(\ln 2))^{d-3}&\\
   &=\delta m^{d-2}(\ln 2)^{d-3}&\\
   \implies \delta_{i-1}m&=\delta m^{d-1}(\ln 2)^{d-3}&\\   &=\delta\left(\ceil{\left(\frac{1}{\delta}\right)^{1/(d-1)}\frac{1}{\ln 2}}\right)^{d-1}(\ln 2)^{d-3}&\\
   &=\frac{1}{\ln 2^{d-1}}(\ln 2)^{d-3}\epsilon&\text{for some $\epsilon\in[1,2]$}\\
   &=\frac{1}{(\ln 2)^2}\epsilon
 \end{flalign*}

 With the above estimate for $\delta_{i-1}m$, we show the required bounds as follows. It follows exactly as in the proof of Part (1) for $i\in [d-2]$ that 
 \[
 p^{(b)}_i = \prob{\bm{x}\in \mu^{N_i}_b}{F_i(\bm{x}) = 1}.
 \]
 Hence, we have
 \begin{flalign*}
   && p_i^{(1)} &= (1-p_{i-1}^{(1)})^{f_{d-1}} &\\
   &&      &= (1-p_{i-1}^{(1)})^{C_1\cdot m2^m} &\\
   &&      &\ge \left(1-\frac{1}{2^m}(1-\delta_{i-1}\exp(-3\delta_{i-1}))\right)^{C_1m2^m} &\text{From part (2)}\\
   &&      &\ge \exp\left(\left(\frac{-1}{2^m}(1-\delta_{i-1}\exp(-3\delta_{i-1}))\right)\left(1+\frac{1}{2^m}(1-\delta_{i-1}\exp(-3\delta_{i-1}))\right)C_1m2^m\right)&\text{By Fact~\ref{fac:exp} (b)}\\
   &&      &\ge \exp\left(\frac{-1}{2^m}\left(1-\delta_{i-1}\exp(-3\delta_{i-1})+\frac{1}{2^m}\right)C_1m2^m\right)&\\
   &&      &= \exp(-C_1m)\exp\left(C_1\left(\delta_{i-1}m\exp(-3\delta_{i-1})-\frac{m}{2^m}\right) \right) &\\
   &&      &\ge \exp(-C_1m) \exp\left(\frac{C_1}{4(\ln 2)^2}\right)&\delta_{i-1}m\geq \frac{1}{(\ln 2)^2}\\
   &&      &\ge \exp(-C_1m) \exp\left(C_2\right)&C_2 = C_1/10\\
   && \text{The }&\text{upper bound for $p_i^{(0)}$ is as follows:}&\\
   && p_i^{(0)}&= (1-p_{i-1}^{(0)})^{C_1m2^m}&\\
   && &\le \left(1-\frac{1}{2^m}(1+\delta_{i-1}\exp(-3\delta_{i-1}))\right)^{C_1m2^m}&\text{From part (2)}\\
   && &\le \exp\left( \frac{-1}{2^m}(1+\delta_{i-1}\exp(-3\delta_{i-1}))C_1m2^m\right)&\text{From Fact~\ref{fac:exp} (c)}\\
   && &\le \exp(-C_1m)\exp(-C_1\delta_{i-1}m\exp(-3\delta_{i-1}))&\\
   && &\le \exp(-C_1m)\exp(-C_1/4(\ln 2)^2)&\\
   && &\le \exp(-C_1m)\exp(-C_2)&\\
 \end{flalign*}

 \textbf{Proof of (4):} Without loss of generality, assume $d\equiv 0\pmod 2$. Then, the output gate of the circuit is an OR gate. Thus:
 \begin{flalign*}
   && \prob{\mu_1^{N_d}}{F_d(\bm{x})=0}&=(1-p_{d-1}^{(1)})^{f_d}&\\
   &&                    &\le \exp(-p_{d-1}^1\cdot f_d)&\text{From Fact~\ref{fac:exp} (c)} \\
   &&                    &\le \exp(-\exp(-C_1m+C_2)\cdot \exp(C_1m))&\text{From part (3)}\\
   &&                    &= \exp(-\exp(C_2))&\\
   &&                    &= \frac{1}{e^{e^5}}\le 0.05&\\
   &&\text{Similarly, }\prob{D_0}{F_d(\bm{x})=1}&\le f_d\cdot \prob{D_n}{F_{d-1}(\bm{x})=1}  &\text{by union bound}\\
   &&                    &\le 2\exp(C_1m)\exp(-C_1m\cdot C_2)&\\
   &&                    &\le 2\exp(-C_2) = \frac{2}{e^5}&\\
   &&                    &\le 0.05&
 \end{flalign*}
\end{proof}

\section{Omitted Proofs from Section \ref{sec:janson}}
\begin{reptheorem}{thm:janson}[Janson's inequality]
Let $C_1,\ldots,C_M$ be any monotone Boolean circuits over inputs $x_1,\ldots,x_N,$ and let $C$ denote $\bigvee_{i\in [M]}C_i.$ For each distinct $i,j\in [M]$, we use $i \sim j$ to denote the fact that $\Vars(C_i)\cap \Vars(C_j)\neq \emptyset$.  Assume each $\bm{x}_j$ ($j\in [n]$) is chosen independently to be $1$ with probability $p_i\in [0,1]$, and that under this distribution, we have $\max_{i\in [M]}\prob{\bm{x}}{C_i(\bm{x}) = 1}\leq 1/2.$ Then, we have
\begin{equation}
\label{eq:janson}
\prod_{i\in [M]}\prob{\bm{x}}{C_i(\bm{x}) = 0} \leq \prob{\bm{x}}{C(\bm{x}) = 0} \leq \left(\prod_{i\in [M]}\prob{\bm{x}}{C_i(\bm{x}) = 0}\right) \cdot \exp(2\Delta)
\end{equation}
where $\Delta := \sum_{i < j: i\sim j} \prob{\bm{x}}{(C_i(\bm{x})=1) \wedge (C_j(\bm{x}) = 1)}.$
\end{reptheorem}

We will use the following inequality in the proof of the above theorem. 

\begin{lemma}[Kleitman's inequality]
  \label{lem:kleitman}
Let $F,G: \{0,1\}^n \rightarrow \{0,1\}$ be two monotonically increasing Boolean functions or monotonically decreasing Boolean functions. Then, 

\[\prob{\bm{x}}{F(\bm{x})=1| G(\bm{x}) =0}  \mathop{\leq}_{\mathrm{(i)}} \prob{\bm{x}}{F(\bm{x})=1} \mathop{\leq}_{\mathrm{(ii)}} \prob{\bm{x}}{F(\bm{x})=1|G(\bm{x})=1} \]
\end{lemma}
\begin{proof}[Proof of Theorem~\ref{thm:janson}]
As $C(\bm{x})$ is an OR over $C_i(\bm{x})$ for $i \in [M]$, $\prob{\bm{x}}{C(\bm{x}) = 0} = \prob{x}{\forall {i \in [M]}, (C_i(\bm{x})=0)}$. The lower bound on $\prob{\bm{x}}{C(\bm{x}) = 0}$ follows easily from Kleitman's inequality (Lemma~\ref{lem:kleitman}) and induction on $M$. 

Now we prove the upper bound  on $\prob{\bm{x}}{C(\bm{x}) = 0}$. In order to prove the intended upper bound, we use the following intermediate lemma. 

\begin{lemma}
  \label{lem:inter}
For all $i \in [M]$, $$\prob{\bm{x}}{(C_i(\bm{x})=1) \mid \forall j < i, (C_j(\bm{x})=0)} \geq \prob{\bm{x}}{C_i(\bm{x})=1} - \sum_{j:j < i, j \sim i} \prob{\bm{x}}{(C_i(\bm{x})=1) ~\mathrm{AND}~ (C_j(\bm{x})=1)}$$
\end{lemma}
Assuming the above lemma, we will complete the proof of Theorem~\ref{thm:janson}. 
\begin{align}
\prob{\bm{x}}{C(\bm{x}) = 0} & = \prob{\bm{x}}{\forall {i \in [M]}, (C_i(\bm{x})=0)} \nonumber \\
 & = \prod_{i \in [M]} \prob{\bm{x}}{(C_i(\bm{x})=0)\mid \forall j < i, (C_j(\bm{x})=0)} \notag \\
& = \prod_{i \in [M]} (1-\prob{\bm{x}}{(C_i(\bm{x})=1)\mid \forall j < i, (C_j(\bm{x})=0)}) \nonumber  \\
& \leq \prod_{i \in [M]} (1-\prob{x}{C_i(\bm{x})=1} - \sum_{j: j < i, j \sim i} \prob{\bm{x}}{(C_i(\bm{x})=1) \text{ AND } (C_j(\bm{x})=1)}) \label{eq:inter} \\
& = \prod_{i \in [M]} (\prob{\bm{x}}{C_i(\bm{x})=0} - \sum_{j: j < i, j \sim i} \prob{\bm{x}}{(C_i(\bm{x})=1) \text{ AND } (C_j(\bm{x})=1)}) \nonumber  \\
& = \prod_{i \in [M]} \left(\prob{\bm{x}}{C_i(\bm{x})=0} \cdot (1+ \frac{1}{\prob{x}{(C_i(\bm{x})=0)}}\sum_{j: j < i, j \sim i} \prob{\bm{x}}{(C_i(\bm{x})=1) \text{ AND } (C_j(\bm{x})=1)}\right) \nonumber \\
& \leq \prod_{i \in [M]} \left(\prob{\bm{x}}{C_i(\bm{x})=0} \cdot (1+ 2\sum_{j: j < i, j \sim i} \prob{\bm{x}}{(C_i(\bm{x})=1) \text{ AND } (C_j(\bm{x})=1)}\right) \label{eq:half} \\
& \leq \prod_{i \in [M]} \left(\prob{\bm{x}}{C_i(\bm{x})=0} \cdot \exp(2 \sum_{j: j < i, j \sim i} \prob{\bm{x}}{(C_i(\bm{x})=1) \text{ AND } (C_j(\bm{x})=1)})\right) \label{eq:expn}\\
& = \left(\prod_{i \in [M]} \prob{\bm{x}}{C_i(\bm{x})=0}\right) \cdot \exp(2\Delta) \nonumber 
\end{align}
Inequality (\ref{eq:inter}) follows from Lemma~\ref{lem:inter}. Inequality (\ref{eq:half}) follows from the fact that $\prob{\bm{x}}{(C_i(\bm{x})=0)} \leq 1/2$ for each $i \in [M]$. Finally, (\ref{eq:expn}) follows from (\ref{eq:exp-ineq}).  Therefore, assuming Lemma~\ref{lem:inter}, we are done. We now prove this lemma.

\begin{proof}[Proof of Lemma~\ref{lem:inter}] By reordindering the indices if required, assume that $d$ is an index such that $d<i$ and for $1\leq j \leq d$, $i \sim j$ and for $d < j < i$, $i \not\sim j$. Let $\mathcal{E}$ be the event that $\left[\forall j \leq d, (C_j(\bm{x})=0)\right]$ and $\mathcal{F}$ be the event that $\left[\forall d < j < i, (C_j(\bm{x})=0)\right]$. 
\begin{align}
&\prob{\bm{x}}{(C_i(\bm{x})=1) \mid  \forall j < i, (C_j(\bm{x})=0)} \nonumber \\
&= \prob{\bm{x}}{(C_i(\bm{x})=1) \mid \mathcal{E} \text{ AND } \mathcal{F}} \nonumber\\ 
& \geq \prob{\bm{x}}{(C_i(\bm{x})=1) \text{ AND }\mathcal{E} \mid \mathcal{F}} & \text{ (Bayes' Rule) } \notag \\
& = \prob{\bm{x}}{(C_i(\bm{x})=1) \mid \mathcal{F}} - \prob{\bm{x}}{(C_i(\bm{x})=1) \text { AND } \overline{\mathcal{E}} \mid \mathcal{F}} \notag \\
& = \prob{\bm{x}}{(C_i(\bm{x})=1)} - \prob{\bm{x}}{(C_i(\bm{x})=1) \text { AND } \exists j \leq d, (C_j(\bm{x})=1) \mid \mathcal{F}}  & (C_i(\bm{x}) \text{ and } \mathcal{F} \text{ are independent}) \notag \\
& = \prob{\bm{x}}{(C_i(\bm{x})=1)} - \prob{\bm{x}}{\exists j \leq d, [(C_i(\bm{x})=1) \text { AND } (C_j(\bm{x})=1)] \mid \mathcal{F}} \nonumber \\
& \geq \prob{\bm{x}}{(C_i(\bm{x})=1)} - \sum_{j \leq d} \prob{\bm{x}}{[(C_i(\bm{x})=1) \text { AND } (C_j(\bm{x})=1)] \mid \mathcal{F}}  & \text{(Union bound)} \notag\\
& \geq \prob{\bm{x}}{(C_i(\bm{x})=1)} - \sum_{j \leq d} \prob{\bm{x}}{[(C_i(\bm{x})=1) \text { AND } (C_j(\bm{x})=1)]} & \text{(Kleitman's inequality)} \notag
\end{align}

\end{proof}
\end{proof}

\end{document}